\documentclass[10pt]{myllncs}
\usepackage{latexsym}
\usepackage{amsmath}
\usepackage{amssymb}
\usepackage{fourier}
\usepackage{fullpage}

\usepackage{color,subfig}
\usepackage{epic}
\usepackage{hyperref}

\usepackage{pgf}
\usepackage{tikz}
\usetikzlibrary{arrows,matrix,shapes,snakes,automata,backgrounds,petri}

\begin{document}

\title{On Three Alternative Characterizations of  Combined Traces}

\author{Dai Tri Man L\^{e}}
\institute{Department of Computer Science,\\
University of Toronto\\
\email{ledt@cs.toronto.edu}}

\maketitle






\newcommand{\LT}{\mathcal{L}}
\newcommand{\eqa}{\thickapprox}
\newcommand{\eqb}{\equiv}
\newcommand{\PS}[1]{\wp(#1)}

\newcommand{\h}[1]{\overline{#1}}
\newcommand{\set}[1]{\{#1\}}
\newcommand{\bset}[1]{\bigl\{#1\bigr\}}
\newcommand{\Bset}[1]{\Bigl\{#1\Bigr\}}
\newcommand{\E}{\mathbb{S}}
\newcommand{\EC}{\widehat{\mathcal{S}}}
\newcommand{\PO}{\prec}
\newcommand{\st}{\mathsf{st}}
\newcommand{\com}{<\!\!>}
\newcommand{\df}{:=}
\newcommand{\iffdf}{\stackrel{\textit{\scriptsize{df}}}{\iff}\ }
\newcommand{\calf}[1]{\mathcal{#1}}
\newcommand{\sq}{\sqsubset}
\newcommand{\todo}[1]{ \textcolor{red}{TODO: #1}}
\newcommand{\tcomment}[1]{\text{\hspace*{2mm}$\langle$~\parbox[t]{\textwidth}{ #1 $\rangle$}}}
\newcommand{\ttcomment}[1]{\text{$\langle$~#1~$\rangle$}}
\newcommand{\imm}[1]{{#1}^{\mathsf{cov}\;}}

\newcommand{\TR}{\mathsf{TR}}
\newcommand{\LCT}{\mathsf{LCT}}
\newcommand{\DCD}{\mathsf{CDG}}
\newcommand{\andspace}{\hspace*{3mm}\text{ and }\hspace*{3mm}}
\newcommand{\myspace}[1]{\mbox{\hspace{#1cm}}}
\newcommand{\map}{\mathsf{imap}}
\newcommand{\bt}{\mathsf{bt}}
\newcommand{\cl}{\mathsf{ct2lct}}
\newcommand{\lc}{\mathsf{lct2ct}}
\newcommand{\cd}{\mathsf{ct2dep}}
\newcommand{\dl}{\mathsf{dep2lct}}
\newcommand{\ld}{\mathsf{lct2dep}}
\newcommand{\enu}{\mathsf{enum}}

\newcommand{\tf}[1]{\mathbf{#1}}

\newcommand{\lhdf}{\lhd^\frown}
\newcommand{\flhd}{\frown_\lhd}
\newcommand{\lex}[1]{\,{#1}^{\textit{lex}}\,}
\newcommand{\optord}{\,<^{\textit{opt}}\,}
\newcommand{\stor}[1]{\,{#1}^{\textit{st}}\,}
\newcommand{\reco}[1]{ {#1}^{\,{\textsf{C}}} }
\newcommand{\si}[1]{(#1)^\Cap}
\newcommand{\CT}[1]{\mathfrak{C}(#1)}
\newcommand{\GCT}[1]{\mathfrak{gC}(#1)}
\newcommand{\eqrel}[1]{\equiv_{#1}}
\newcommand{\quotient}[2]{{#1}/{#2}}
\newcommand{\FEC}{\mathbb{C}\mathit{at}_{\mathsf{FE}}}
\newcommand{\cau}{\longrightarrow}
\newcommand{\wcau}{\dasharrow}
\newcommand{\bcau}{\rightsquigarrow}

\newcommand{\seq}[1]{\left[   #1 \right] }
\renewcommand{\emptyset}{\varnothing}

\newcommand{\It}[1]{\mathit{#1}}
\newcommand{\defref}[1]{Definition~\ref{def:#1}}
\newcommand{\theoref}[1]{Theorem~\ref{theo:#1}}
\newcommand{\propref}[1]{Proposition~\ref{prop:#1}}
\newcommand{\corref}[1]{Corollary~\ref{cor:#1}}
\newcommand{\lemref}[1]{Lemma~\ref{lem:#1}}
\newcommand{\exref}[1]{Example~\ref{ex:#1}}
\newcommand{\reref}[1]{Remark~\ref{re:#1}}
\newcommand{\eref}[1]{\eqref{eq:#1}}
\newcommand{\figref}[1]{Figure~\ref{fig:#1}}

\newcommand{\secref}[1]{Section~\ref{sec:#1}}
\newcommand{\eqnref}[1]{Eq.~(\ref{eq:#1})}
\newcommand{\TCT}{\stackrel{\mathsf{t}\leftrightsquigarrow \mathsf{c}}{\equiv}}

\newcommand{\EOD}{\hfill {$\blacktriangleleft$}}

\numberwithin{equation}{section}


\begin{abstract}
The \emph{com}bined \emph{trace}  (i.e., \emph{comtrace}) notion was introduced by Janicki and Koutny in 1995 as a generalization of the Mazurkiewicz trace notion. Comtraces are congruence classes of step sequences, where the congruence relation is defined from two relations \emph{simultaneity} and \emph{serializability} on events.  They also showed that comtraces correspond to \emph{some} class of labeled stratified order structures, but left open the question of  what class of labeled stratified orders represents comtraces. In this work, we proposed a class of labeled stratified order structures that captures exactly the  comtrace notion. Our main technical contributions are representation theorems showing that comtrace quotient monoid,  \emph{combined dependency graph} (Kleijn and Koutny 2008) and our labeled stratified order structure characterization are three different and yet equivalent ways to represent comtraces. This paper is a revised and expanded version of  L\^e (in Proceedings of PETRI NETS 2010, LNCS 6128, pp. 104-124).\\

\keywordname\quad causality theory of concurrency, generalized trace theory, combined trace, step sequence, stratified order structure
\end{abstract}


\section{Introduction}
Partial orders are one of the main tools for modelling ``true concurrency'' semantics of concurrent systems (cf. \cite{Pra}). They are utilized to develop powerful partial-order based automatic verification techniques, e.g.,  the partial order reduction  technique for model checking of concurrent software (see, e.g., \cite[Chapter 10]{CGP} and \cite{EH}). Partial orders are also equipped with \emph{traces}, their powerful formal language counterpart,  introduced by Mazurkiewicz in his seminal paper \cite{Ma1}. In  \emph{The Book of Traces} \cite{DR}, trace theory has been used to tackle problems from diverse areas including formal language theory, combinatorics, graph theory, algebra, logic, and especially concurrency theory.

However while partial orders and traces can sufficiently model the ``earlier than" relationship, Janicki and Koutny argued that it is problematic to use a single partial order to specify both  the ``earlier than" and  the ``not later than" relationships \cite{J4}. This motivates them to develop the theory of \emph{relational structures}, where a pair of relations is used to capture concurrent behaviors.  The most well-known among the classes of relational structures  proposed by Janicki and Koutny is the class of \emph{stratified order structures} (\emph{so-structures}) \cite{GP,JK0,JK97}. A so-structure is a triple $(X,\prec,\sqsubset)$, where $\prec$ and $\sqsubset$ are binary relations on $X$. They were invented to model both the ``earlier than" 
(the relation $\prec$) and ``not later than" (the relation $\sqsubset$) relationships, under the assumption that system runs can be described using \emph{stratified partial orders}, i.e., step sequences. So-structures have been successfully used to give semantics of inhibitor and priority systems \cite{JK99,KK,JLM08,JLM08b}. 

The \emph{com}bined \emph{trace} (\emph{comtrace}) notion, introduced by Janicki and Koutny \cite{JK95}, generalizes the trace notion by utilizing step sequences instead of words. First the set of all possible steps that generates step sequences are identified by a relation $sim$, which is called {\em simultaneity}. Second a congruence relation is determined by a relation $ser$, which is called {\em serializability} and is in general \emph{not} symmetric. Then a comtrace is defined to be a  congruence class of step sequences. Comtraces were introduced as a formal language representation of so-structures to provide an operational semantics for  Petri nets with inhibitor arcs. Unfortunately comtraces have been less often known and applied than so-structures, even though in many cases they appear to be more natural. We believe one  reason is that the comtrace notion was too succinctly discussed in \cite{JK95} without a full treatment dedicated to comtrace theory. Motivated by this, Janicki and the author have devoted our recent effort on the study of comtraces \cite{Le08,JL11}, yet there are too many different aspects to explore and the truth is  that we could barely scratch the surface. In particular, a huge amount of results from  trace theory (e.g., from \cite{DR,DM}) need to be generalized to comtraces. These  tasks are often challenging since we are required to develop novel techniques to deal with the complex interactions of the ``earlier than'' and ``not later than'' relations.

\subsection{Motivation}
In the literature of Mazurkiewicz traces, traces are defined using the following three equivalent methods. The first method is to define a trace to be a congruence class of words, where the congruence relation is induced from an independency relation on events. In the second method, a trace can be viewed as  a \emph{dependence graph} (cf. \cite[Chapter 2]{DR}). A dependence graph is a directed acyclic  graph whose vertices are labeled with events, and satisfies the condition that every two distinct vertices with dependent labels must be connected by exactly one directed edge.  The third method is to define a trace as a labeled partially ordered set, whose elements are labeled with events,  and we also require  the partial order  to be ``compatible'' with the independency relation (see \defref{ltraces} for the precise formulation). Although the above three characterizations of traces can be shown to be equivalent, depending on the situation one characterization can be more convenient than the others. When studying graph-theoretics of traces, the dependence graph representation is the most natural. The treatment of traces as congruence classes of words is more convenient in Ochma\'nski's characterization of recognizable trace languages \cite[Chapter 6]{DR} and Zielonka's theory of asynchronous automata \cite[Chapter 7]{DR}. On the other hand most results on temporal logics for traces (see, e.g., \cite{TW02,Wal02,DG06a,DG06b,GK07}) utilize the labeled poset representation of traces. The reason is that it is more natural to interpret temporal logics on the vertices (for local temporal logics) or \emph{finite downward closed subsets} (for global temporal logics) of a labeled partially order set. Thus all of these three representations are indispensable in Mazurkiewicz trace theory.

Since our long-term goal is to generalize the results from  Mazurkiewicz trace theory to comtraces, there is a strong need for all three analogous representations for comtraces.  In \cite{JK95}, Janicki and Koutny  already gave us the definition of comtraces using congruence classes of step sequences. They also showed that every comtrace can be represented by a labeled so-structure, but  a direct method for defining comtraces using labeled so-structures was not given. Inspired by the  dependence graph representation of traces,  Kleijn and Koutny \cite{KK08}  recently introduced the \emph{combined dependency graph} (\emph{cd-graph}) notion, but a theorem showing that cd-graphs can be represented by comtraces was not given. Thus the goal of this paper is to complete the picture by giving  a  new characterization of comtraces using labeled so-structures and developed a unified framework to show the equivalence of these three representations of comtraces.

\subsection{Organization}
This paper is the revised and expanded version of the conference paper \cite{Le10}. Although no new results are added, several sections and  proofs are rewritten and more examples are included to improve readability of this paper. We also fix a few serious typos and mistakes found in the previous version. The paper is organized as follows.

In \secref{background}, we recall some preliminary definitions and notations. In \secref{survey}, we gives a concise exposition of the theory of so-structures and comtraces \cite{JK95,JK97} to make the paper self-contained. 

In \secref{lsos-comtrace}, we introduce the concept of quotient so-structure and use it to construct our definition of comtraces using labeled so-structure, which we called \emph{lsos-comtraces} to avoid confusion with the original comtrace notion by Janicki and Koutny in \cite{JK95}. We also give some remarks on how we arrived  at such definition. 

\secref{representation} contains the main technical contributions. We prove the first representation theorem which establishes bijections between the set of all comtraces and the set of all lsos-comtraces over the same comtrace alphabet. Then using this theorem, we prove the second representation theorem which provides bijections between the set of all lsos-comtraces and the set of all cd-graphs \cite{KK08} over the same alphabet. 

In \secref{op}, we define  \emph{composition} operators for lsos-comtraces and for cd-graphs, analogous to the comtrace concatenation operator, and show that the set of all lsos-comtraces (or cd-graphs) over a fixed comtrace alphabet together with its composition operator forms a monoid. We also strengthen the representation theorems from Section 5 by showing that the bijections from these theorems are indeed monoid isomorphisms.

Finally, \secref{conc} contains some final remarks and future work.

\section{Notation \label{sec:background}}

\subsection{Relations, orders and equivalences}

The \emph{powerset} of a set $X$ will be denoted by $\PS{X}$. We let $\emph{id}_X$ denote the \emph{identity relation} on a set $X$. We write $R\circ S$ to denote the \emph{composition} of relations $R$ and $S$. We write $R^{*}$ to denote the \emph{reflexive transitive closure} of $R$ respectively.

Let $f:A\rightarrow B$ be a function, then for every set $C\subseteq A$, we write $f[C]$ to denote the image of the set $C$ under $f$, i.e., $f[C] \df \set{f(x) \mid x\in C}$.

A binary relation $R\subseteq X\times X$ is an \emph{equivalence relation} relation on $X$ if and only if it is reflexive, symmetric and transitive. If $R$ is an equivalence relation, we write $[x]_R$ to denote the equivalence class of $x$ with respect to $R$, and the set of all equivalence classes in $X$  is denoted as $X/R$ and called the \emph{quotient set} of $X$ by $R$.  We drop the subscript and write $[x]$ to denote the equivalence class of $x$ when $R$ is clear from the context.

A binary relation $\prec \;\subseteq X \times X$ is a
{\em partial order} if and only if it is {\em irreflexive} and {\em transitive}.
The pair $(X,\prec)$ in this case is called a \emph{partially ordered set} (\emph{poset}). The pair $(X,\prec)$ is called a \emph{finite poset} if $X$ is finite. For convenience, we define:
\begin{align*}
\simeq_\prec  &\df  \bset{(a,b)\in X\times X\mid a\not\prec b \;\wedge\; b\not\prec a}& \text{\emph{(incomparable)}}\\
\frown_\prec  &\df \bset{(a,b)\in X\times X\mid  a \simeq_\prec b \;\wedge\; a\neq b} &
\text{\emph{(distinctly incomparable)}}\\
\prec^\frown  &\df \bset{(a,b)\in X\times X\mid a \prec b \;\vee\; a\frown_\prec b} & \text{\emph{(not greater)}}
\end{align*}

A poset $(X,\prec)$ is {\em total} if and only if $\frown_\prec$ is empty; and {\em stratified} if and only if $\simeq_\prec$ is an equivalence relation. Evidently every total order is stratified.

\subsection{Step sequences\label{sec:steps}}
For every finite set $E$, a set $\E \subseteq \PS{E} \setminus \set{\emptyset}$ can be seen as an alphabet. 
The elements of $\E$ are called {\em steps} and
the elements of $\E^*$ are called {\em step sequences}. For example, if the set of possible steps is 
$\E = \bset{ \{a,b,c\}, \{a,b\}, \{a\}, \{c\} }$, then
$\{a,b\}\{c\}\{a,b,c\} \in \E^*$ is a step sequence.
The triple $(\E^*,\ast,\epsilon)$, where $\_\,\ast\,\_$ denotes the step sequence concatenation operator (usually omitted) and $\epsilon$ denotes the empty step sequence, is a monoid.

Let $t=A_1\ldots A_k$ be a step sequence. We define  $|t|_a$, the number of occurrences of an event $a$ in $t$, as $|t|_a \df \bigl\lvert\bset{A_i\mid 1\le i \le k \wedge a\in A_i}\bigr\rvert$. Then we can construct its unique \emph{enumerated step sequence} $\h{t}$ as\\
\[\h{t}\df \h{A_1}\ldots \h{A_k}\text{, where } \h{A_i} \df \Bset{e^{(|A_1\ldots A_{i-1}|_e+1)}\bigl\lvert e\in A_i\bigr.}.\]
We will call such $\alpha=e^{(j)}\in \h{A_i} $ an \emph{event occurrence} of $e$. We let $\Sigma_t=\bigcup_{i=1}^k \h{A_i}$ denote the set of all event occurrences in all steps of $t$.   We also define $\ell:\Sigma_t\rightarrow E$ to be the function that returns the \emph{label} of $\alpha$ for each $\alpha \in \Sigma_t$. For example, if $\alpha=e^{(j)}$, then $\ell(\alpha)=\ell(e^{(j)})=e$.

\begin{example}\label{ex:ss1}
Given the step sequence   $t= \set{a,b}\set{b,c}\set{c,a}\set{a}$, 
then the enumerated step sequence is  
\begin{align*}
\overline{t}=\bset{a^{(1)},b^{(1)}}\bset{b^{(2)},c^{(1)}}\bset{a^{(2)},c^{(2)}}\bset{a^{(3)}}.
\label{eq:t2}
\end{align*}
The set of event occurrences is 
$\Sigma_t=\bset{ a^{(1)},a^{(2)},a^{(3)}, b^{(1)}, b^{(2)},
c^{(1)},c^{(2)} }.$
\EOD
\end{example}

Given a step sequence $s=B_{1}\ldots B_{m}$ and any function $f$ defined on $\bigcup_{i=1}^m B_i$, we define \[\map(f,s)\df f[B_{1}]\ldots f[B_{m}],\] which is  the step sequence derived from $s$ by computing the image of  each $A_{i}$ under the function $f$ successively. Using the function $\map$, from an enumerated step sequence $\h{t} = \h{A_1}\ldots \h{A_k}$, we can recover its step sequence $t =  \map(\ell, t)=  \ell[\,\h{A_1}\,]\ldots \ell[\,\h{A_k}\,].$

For each $\alpha \in \Sigma_t$, we let $pos_t(\alpha)$ denote the consecutive number of a step where $\alpha$ belongs, i.e., if $\alpha\in \h{A_i}$ then $pos_t(\alpha)=i$. For our example, $pos_t(a^{(2)})=3$ and  $pos_t(b^{(2)})=pos_t(c^{(1)})=2$. 

It is important to note that step sequences and stratified orders are interchangeable concepts. Given a step sequence $t$, define the binary relation $\lhd_t$  on $\Sigma_t$ as
\[\alpha \lhd_t \beta \iffdf pos_u(\alpha)<pos_t(\beta).\]

Intuitively, $\alpha \lhd_t \beta$ simply means $\alpha$ occurs before $\beta$ on the step sequence $u$. Thus, $\alpha \lhd_t^\frown \beta$ if and only if $(\alpha\not=\beta\wedge pos_t(\alpha)\le pos_t(\beta))$; and $\alpha\simeq_t \beta$ if and only if  $pos_t(\alpha)=pos_t(\beta)$. Obviously, the order $\lhd_t$ is stratified and we will call it the stratified order generated by the step sequence $t$. For instance, from the step sequence $t$ and its enumerated step sequence $\h{t}$ as given in  \exref{ss1}, the stratified order $\lhd_{t}$ is shown in \figref{so1} (the edges that can be inferred by transitivity are omitted).

\begin{figure} 
\begin{center}
\begin{footnotesize}
\begin{tikzpicture}[scale=1.5,->,>=stealth']
        \tikzstyle{every node}= [circle,fill=black!30,minimum size=20pt,inner sep=0pt]
 
	\foreach \pos/\name/\index in {{(0,0)/a/1}, {(0,-1)/b/1}, {(1,0)/b/2},
                            {(1,-1)/c/1}, {(2,0)/a/2}, {(2,-1)/c/2}, {(3,-0.5)/a/3}}
        \node (\name\index) at \pos {$\name^{(\index)}$};
        
        \foreach \sou/\des in { a1/b2, a1/c1,  b1/b2, b1/c1, b2/a2, b2/c2, c1/a2, c1/c2,a2/a3,c2/a3}
        		\path [thick,shorten >=1pt]  (\sou) edge (\des);
\end{tikzpicture}	
\end{footnotesize}
\end{center}
\caption{The stratified order $\lhd_{t}$ defined from the step sequence $t$ given in \exref{ss1}.}
\label{fig:so1}
\end{figure}

Conversely, let $\lhd$ be a stratified order on a set $\Sigma$. The set $\Sigma$ can be partitioned into a
sequence of equivalence classes $\Omega_\lhd=B_1\ldots B_k$ ($k\ge 0$) such that \[\lhd = \bigcup_{i<j}B_i\times B_j\text{\mbox{\hspace{4mm}} and \mbox{\hspace{4mm}} }\simeq_\lhd \;= \bigcup_{i}B_i\times B_i.\]
The sequence $\Omega_\lhd$ is called the step sequence \emph{representing} $\lhd$. For example, the let $\lhd$ be the stratified order in \figref{so1}, then 
\[\Omega_\lhd = \bset{a^{(1)},b^{(1)}}\bset{b^{(2)},c^{(1)}}\bset{a^{(2)},c^{(2)}}\bset{a^{(3)}},\]
which is exactly the enumerated step sequence $\h{t}$ given in \exref{ss1}. To get back the step sequence $t$, we need to get rid of all the superscripts  as follows:
\begin{align*}
\map(\ell,\h{t}\,) &= \ell\left[\bset{a^{(1)},b^{(1)}}\right]\ell\left[\bset{b^{(2)},c^{(1)}}\right]\ell\left[\bset{a^{(2)},c^{(2)}}\right]\ell\left[\bset{a^{(3)}}\right]\\
&=  \set{a,b}\set{b,c}\set{c,a}\set{a}
=  t
\end{align*}

\section{Stratified order structures and combined traces \label{sec:survey}}
In this section, we review the Janicki-Koutny theory of stratified order structures and comtraces from \cite{JK95,JK97}. This introduction might be too concise for readers who are not familiar with the subject, so we also refer to \cite{KK08} for an excellent introductory tutorial on traces, comtraces, partial orders and so-structures with many motivating examples.

\subsection{Stratified order structures}
A \emph{relational structure} is a triple $T=(X,R_1,R_2)$, where $X$ is a set and $R_1$, $R_2$ are binary relations on $X$. A relational structure $T'=(X',R'_1,R'_2)$ is an \emph{extension} of $T$, denoted as $T\subseteq T'$, if and only if $X=X'$, $R_1\subseteq R_1'$ and $R_2\subseteq R_2'$.

\begin{definition}[stratified order structure \cite{JK97}] \label{def:sos}
A \emph{stratified order structure} (\emph{so-structure}) is a relational structure $S=(X,\prec,\sqsubset),$ such that for all $\alpha,\beta,\gamma \in X$, the following hold:
\begin{align*}
\text{\textsf{S1:}\hspace{3mm}}& \alpha\not\sqsubset \alpha & \text{\textsf{S3:}\hspace{3mm}}& \alpha\sqsubset \beta \sqsubset \gamma \;\wedge\; \alpha \not= \gamma \implies \alpha\sqsubset \gamma\\
\text{\textsf{S2:}\hspace{3mm}}& \alpha \prec \beta \implies \alpha \sqsubset \beta \hspace{3mm}& \text{\textsf{S4:}\hspace{3mm}}&
\alpha\sqsubset \beta \prec \gamma \;\vee\; \alpha\prec \beta \sqsubset \gamma \implies \alpha\prec \gamma
\end{align*}
When $X$ is finite, $S$ is called a \emph{finite so-structure}. \EOD
\end{definition}

The axioms \textsf{S1}--\textsf{S4} imply that $\PO$ is a partial order and $\alpha\prec \beta\Rightarrow \beta\not\sq \alpha.$  The axioms \textsf{S1} and \textsf{S3} imply $\sq$ is a \emph{strict preorder}. The relation $\prec$ is called \textit{causality} and represents the ``earlier than" relationship, while the relation $\sqsubset$ is called \textit{weak causality} and represents the ``not later than" relationship. The axioms \textsf{S1}-\textsf{S4} model the mutual relationship between ``earlier than" and ``not later than" relations, provided that {\em system runs are modeled by stratified orders}.  Historically, the name ``stratified order structure'' came from the fact that stratified orders can be seen as a special kind of so-structures.

\begin{proposition}[cf. \cite{J4}] \label{prop:soss}
For every stratified poset $(X,\lhd)$, the triple $S_\lhd=(X,\lhd,\lhd^\frown)$ is an so-structure.
\end{proposition}

We next recall the notion of \emph{stratified order extension}. This concept is important for our purpose since the relationship between stratified orders and so-structures is analogous to the one between total orders and partial orders.

\begin{definition}[stratified extension \cite{JK97}] 
Let $S=(X,\prec,\sqsubset)$ be an so-structure. A {\em stratified} order $\lhd$ on $X$ is a {\em stratified extension} of $S$ if and only if $(X,\prec,\sqsubset)\subseteq (X,\lhd,\lhd^{\frown})$. The set of all stratified extensions of $S$  is denoted as  $ext(S)$.  \EOD
\label{def:extsos}
\end{definition}

Szpilrajn's Theorem \cite{Szp} states that every poset can be reconstructed by taking the intersection of all of its total order extensions. Janicki and Koutny showed that a similar result holds for so-structures and stratified extensions.

\begin{theorem}[{\cite{JK97}}] \label{theo:SzpStrat}
Let $S=(X,\PO,\sq)$ be an so-structure. Then
\[S=\left(X,\bigcap_{\lhd\;\in\; ext(S)}\lhd,\bigcap_{\lhd\;\in\; ext(S)}\lhd^\frown\right).\]
\end{theorem}

This theorem holds even when $X$ is infinite, and its proof requires some version of the axiom of choice. But we are only concerned with finite so-structures in this paper.  Using this theorem, we can show the following properties relating so-structures with their stratified extensions.
\begin{corollary} \label{cor:SzpStrat}
For every so-structure $S=(X,\PO,\sq)$,
\begin{enumerate}
 \item $\bigl(\exists \lhd\in ext(S),\ \alpha\lhd\beta\bigr)\wedge\bigl(\exists \lhd\in ext(S),\ \beta\lhd\alpha\bigr)\implies \bigl(\exists \lhd\in ext(S),\ \beta\frown_\lhd\alpha\bigr).$
 \item $\bigl(\forall \lhd\in ext(S),\ \alpha\lhd\beta \vee \beta\lhd\alpha\bigr)\iff \alpha \PO \beta \vee \beta \PO \alpha.$
\end{enumerate}
\end{corollary}
\begin{proof}\textbf{1. } See \cite[Theorem 3.6]{JK97}. \\
\textbf{2. } Follows from \textbf{1.} and \theoref{SzpStrat}. 
\qed
\end{proof}

\subsection{Combined traces}

{\em Comtraces} (\emph{Com}bined \emph{traces}) were introduced in \cite{JK95} as a generalization of traces to represent so-structures. The \emph{comtrace congruence} is defined via two binary relations {\em simultaneity} and {\em serializability} on a finite set of events.

\begin{definition}[comtrace alphabet \cite{JK95}] \label{def:comalpha}
Let $E$ be a finite set of events, and  $ser \subseteq sim \subset E\times E$  two relations called \emph{serializability} and \emph{simultaneity} respectively, where $sim$ is irreflexive and symmetric. The triple $\theta = (E,sim,ser)$ is called a \emph{comtrace alphabet}.  \EOD
\end{definition}

Note that since $sim$ is irreflexive and $ser\subseteq sim$, it follows that the relation $ser$ is also irreflexive. Intuitively, if $(a,b)\in sim$ then $a$ and $b$ may occur simultaneously with each other in a step $\set{a,b}$. If $(a,b)\in ser$, then $a$ and $b$ may occur together in a step $\set{a,b}$ and, moreover,  such step can be split into the sequence $\set{a}\set{b}$. A step can involve more than two events as long as all events within the same step are pairwise related by the simultaneity relation. More formally, we define the set of all possible  steps $\E_\theta$ induced from the simultaneity relation $sim$ to be the set of all cliques of
the graph $(E,sim)$, i.e.,
\[\E_{\theta} \df \bset{ A \mid A\neq\emptyset \;\wedge\; \forall a,b\in A, \bigl(a=b \vee (a,b)\in sim\bigr)}.\]

\begin{definition}[comtrace congruence \cite{JK95}] \label{def:commonoid}
For a comtrace alphabet $\theta=(E,sim,ser)$, we define $\eqa_{\theta}\;\subseteq\; \E_{\theta}^*\times\E_{\theta}^*$ to be the relation comprising all pairs $(t,u)$ of step sequences such that 
\begin{align*}
t&=wAz\\ 
u&=wBCz
\end{align*}
where $w,z\in\E_{\theta}^*$ and $A$, $B$, $C$ are steps in $\E_{\theta}$ satisfying $B\cup C \;= \;A$ and $B\times C\;\subseteq\; ser$. 

We define \emph{comtrace congruence} $\eqb_{\theta}\;\df \left(\eqa_{\theta}\cup\eqa^{-1}_{\theta}\right)^*$, and the equivalence classes in $\E^*_{\theta}/\!\!\equiv_{\theta}$ are called \emph{comtraces}. 

We define the comtrace concatenation operator $\_\circledast\_$ as $[r]\circledast[t] \df [r\ast t]$. The quotient monoid $(\E^*_{\theta}/\!\!\equiv_{\theta},\circledast,[\epsilon])$ is called the comtrace monoid over the comtrace alphabet $\theta$.  \EOD
\end{definition}

Note that since $ser$ is irreflexive, $B\times C\subseteq ser$ implies that $B\cap C=\emptyset$. The fact that the comtrace concatenation operator $\_\circledast\_$ is well-defined was shown in \cite[Proposition 4.14]{JK95}.  We will omit the subscript $\theta$ from $\eqa_{\theta}$ and $\eqb_{\theta}$, and write $\equiv$ and $\eqa$ when it causes no ambiguity. To shorten our notations, we often write $[s]_{\theta}$ or $[s]$ instead of $[s]_{\eqb_{\theta}}$ to denote the comtrace generated by the step sequence $s$ over the comtrace alphabet $\theta$.

\begin{example} \label{ex:comtrace1}
Consider three  atomic operations $a$, $b$ and $c$ as follows
\[a:\;\; y\leftarrow x+y \myspace{1,5} b:\;\; x\leftarrow y+2\myspace{1.5}  c:\;\; y \leftarrow y+1\]
Assume simultaneous reading is allowed, but simultaneous writing is not allowed. Then the events $b$ and $c$ can be performed simultaneously, and the execution of the step $\set{b,c}$ gives the same outcome as executing $b$ followed by $c$. The events $a$ and $b$ can also be performed simultaneously, but the outcome of executing the step $\set{a,b}$ is not the same as executing $a$ followed by $b$, or $b$ followed by $a$. Note that although executing the steps $\set{a,b}$ and $\set{b,c}$ is allowed, we cannot execute the step $\set{a,c}$ since that would require writing on the same variable $y$. Thus, the simultaneity relation $sim$ is not transitive. 

Let $E=\set{a,b,c}$ be the set of events. Then we can define the comtrace alphabet $\theta=(E,sim,ser)$, where $sim = \bset{(a,b),(b,a),(b,c),(c,b)}$ and $ser=\set{(b,c)}$. Thus the set of all possible steps is \[\E_{\theta}=\bset{\{a\},\{b\},\{c\},\set{a,b},\set{b,c}}.\]
We observe that the set $\textbf{t}=[\{a\}\set{a,b}\{b,c\}] =\bset{ \{a\}\set{a,b}\{b,c\},\{a\}\set{a,b}\{b\}\{c\}}$ is a comtrace.
But the step sequence $\{a\}\set{a,b}\{c\}\{b\}$ is not an element of $\textbf{t}$ because $(c,b)\not\in ser$.  \EOD
\end{example}

Even though traces correspond to quotient monoids over sequences and comtraces correspond to
quotient monoids over step sequences, traces can be regarded as a special kind
of comtraces when the relation $ser=sim$. For a more detailed discussion on this connection between traces and comtraces, the reader is referred to \cite{JL11}.

\begin{definition}[\cite{JK95}] \label{def:s2inv}
Let $u\in \E_{\theta}^*$. We define the relations $\prec_u,\sqsubset_u \subseteq \Sigma_{u}\times \Sigma_{u}$ as
\begin{enumerate}
\item
$\alpha \prec_u \beta \iffdf \alpha \lhd_u \beta \wedge (\ell(\alpha),\ell(\beta))\notin ser,$

\item
$ \alpha \sqsubset_u \beta \iffdf \alpha \lhd_u^\frown \beta \wedge (\ell(\beta),\ell(\alpha))\notin ser$.  \EOD
\end{enumerate}
\end{definition}

It is worth noting that the structure $( \Sigma_{u}, \prec_u ,\sqsubset_u, \ell )$ is exactly the \emph{cd-graph} (cf. \defref{comdag}) that represents the comtrace $[u]$. This gives us some intuition on how Kleijn and Koutny  constructed their cd-graph definition in \cite{KK08}. The structure $( \Sigma_{u}, \prec_u ,\sqsubset_u )$ is usually \emph{not} an so-structure since $\PO_u$ and $\sq_u$ describe only basic ``local'' causality and weak causality invariants of the event occurrences of $u$ by considering pairwise serializable relationships of event occurrences, and thus $\PO_u$ and $\sq_u$ might not capture ``global'' invariants that can be inferred from \textsf{S2}-\textsf{S4} of \defref{sos}. To ensure  all invariants are included, we need the following $\lozenge$-closure operator. 

\begin{definition}[\cite{JK95}]\label{def:SO-CL}
For every relational structure $S=(X,R_1,R_2)$, we define $S^\lozenge$ as
\[S^\lozenge \df \bigl(X,(R_1\cup R_2)^* \circ R_1 \circ (R_1\cup R_2)^*,(R_1\cup R_2)^* \setminus id_X\bigr).\]  
\EOD
\end{definition}

Intuitively the $\lozenge$-closure generalizes the transitive closure for relations to relational structures. The motivation is that for appropriate  relations $R_1$ and $R_2$ (see assertion (3) of \propref{so-cl} below), the relational structure $(X,R_1,R_2)^\lozenge$ is an so-structure. The $\lozenge$-closure satisfies the following properties.

\begin{proposition}[\cite{JK95}] \label{prop:so-cl}
Let $S=(X,R_1,R_2)$ be a relational structure.
\begin{enumerate}
\item If $R_2$ is irreflexive then $S\subseteq S^\lozenge$.
\item $(S^\lozenge)^\lozenge = S^\lozenge$.
\item $S^\lozenge$ is an so-structure if and only if 
$(R_1\cup R_2)^* \circ R_1 \circ (R_1\cup R_2)^*$
is irreflexive.
\item If $S$ is an so-structure then $S=S^\lozenge$.
\item If $S$ is an so-structure and $S_0 \subseteq S$, then $S_0^{\lozenge}\subseteq S$ and $S_0^{\lozenge}$ is an so-structure. \hspace{-5mm}
\end{enumerate}
\end{proposition}

\begin{definition}\label{def:s2sos}
Given a step sequence  $u\in \E_{\theta}^*$ and its respective comtrace $\mathbf{t}=[u]\in \E_{\theta}^*/\!\equiv$, we define the relational structure $S_{\mathbf{t}}$ as: 
\[S_{\mathbf{t}}  = \bigl( \Sigma_{\mathbf{t}}, \prec_{\mathbf{t}} ,\sqsubset_{\mathbf{t}} \bigr)\df \bigl( \Sigma_{u}, \prec_u ,\sqsubset_u \bigr)^\lozenge.\]
\EOD
\end{definition}

The relational structure $S_{\mathbf{t}}$ is called the \emph{so-structure defined by the comtrace} $\mathbf{t}=[u]$, where $\Sigma_{\mathbf{t}}$, $\prec_{\mathbf{t}}$  and $\sqsubset_{\mathbf{t}}$ are used to denote the event occurrence set, causality relation and weak causality relation induced by the comtrace $\mathbf{t}$ respectively.  The following nontrivial theorem and its corollary justifies the name by showing that  step sequences in a comtrace $\mathbf{t}$ are exactly the stratified extensions  of the so-structure $S_\mathbf{t}$, and that $S_{\mathbf{t}}$ is uniquely defined for the comtrace $\mathbf{t}$ regardless of the choice of the step sequence $u\in \mathbf{t}$. 

\begin{theorem}[\cite{JK95}] \label{theo:com2sos}
For each comtrace $\mathbf{t}\in \quotient{\E^*_{\theta}}{\equiv}$, the structure $S_{\mathbf{t}}$ is an so-structure and the set $ext\bigl(S_{\mathbf{t}}\bigr)$ of stratified extesions of $S_{\mathbf{t}}$ is exactly the same as $\bset{ \lhd_u \mid u \in \mathbf{t} }$, the set  of stratified orders induced by the step sequences in the comtrace $\mathbf{t}$. 
\end{theorem}

\begin{corollary} \label{cor:com2sos}
For all comtraces $\mathbf{t},\mathbf{q}\in \quotient{\E^*_{\theta}}{\equiv}$,
\begin{enumerate}
\item $\mathbf{t}=\mathbf{q}$ if and only if $S_{\mathbf{t}}=S_{\mathbf{q}}$ 
\item $ S_{\mathbf{t}}  = \bigl( \Sigma_{\mathbf{t}}, \prec_{\mathbf{t}} ,\sqsubset_{\mathbf{t}} \bigr) = \Bigl( \Sigma_{\mathbf{t}}, \bigcap_{w\in \mathbf{t}}\lhd_{w} ,\bigcap_{w\in \mathbf{t}}\lhd_{w}^{\frown} \Bigr)$ 
\end{enumerate}
\end{corollary}

 The first part of the corollary states that two comtraces are the same if and only if they define the same so-structure. The second part of corollary gives us two equivalent methods for constructing the so-structure from the comtrace $\mathbf{t}$. The first method is to  use the construction from \defref{s2sos}. The second method is to consider all the step sequences in $\mathbf{t}$, and for every two event occurrences $\alpha$ and $\beta$ of  $\mathbf{t}$, define $\alpha \PO_{\mathbf{t}} \beta$ if $\alpha$ always occurs strictly before $\beta$ in all of these step sequences, and define $\alpha \sq_{\mathbf{t}} \beta$ if $\alpha$ always occurs before or simultaneously with $\beta$ in all of these step sequences.

\section{Comtraces as labeled stratified order structures \label{sec:lsos-comtrace}}
Even though \theoref{com2sos} shows that each comtrace can be represented uniquely by a labeled so-structure, it does not give us any explicit definition describing how these labeled so-structures look like. The goal of this section is to define exactly the class of labeled so-structures  representing comtraces. To provide us with more intuition, we will first recall how Mazurkiewicz traces can be characterized using labeled posets.

A \emph{trace alphabet}  is  a pair $(E,ind)$, where
$ind$ is a symmetric irreflexive binary relation on the finite set $E$. A \emph{trace congruence} $\equiv_{ind}$ can then be defined as the smallest  equivalence relation such that for all sequences $uabv, ubav \in E^*$, if $(a,b) \in \textit{ind}$, then $uabv \equiv_{ind} ubav$. The elements of $\quotient{E^{*}}{\equiv_{ind}}$ are called \emph{traces}.  

Traces can also be defined alternatively as posets whose elements are labeled with symbols of a concurrent alphabet $(E,ind)$ satisfying certain conditions.  

Given a binary relation $R\subseteq X\times X$, the \emph{covering relation} of $R$ is defined as \[\imm{R}\df \bset{(x,y)\mid x\; R\; y \wedge \neg \exists z,\; x\; R\; z\; R \;y }.\]
In other words, $x\,\imm{R} y$ if and only if  $x$ is \emph{immediately related} to $y$ by the relation $R$. When $R$ is a partial order, the graph representing $\imm{R}$ is exactly the familiar Hasse diagram of $R$, and the notion of covering relation is easier to visualize in this case.

Using the covering relation notion, an alternative definition of Mazurkiewicz trace can be given as in the following definition.

\begin{definition}[cf. \cite{TW02}] \label{def:ltraces}
A trace over a trace alphabet $(E,ind)$ is a finite labeled poset $P=(X,\PO,\lambda)$, where $\lambda:X\rightarrow E$ is a labeling function, such that for all elements $\alpha\not=\beta$ of $X$,
\begin{enumerate}
\item $\alpha \imm{\PO} \beta \implies (\lambda(\alpha),\lambda(\beta))\not \in ind$ (immediately causally related event occurrences  must be labeled with dependent events), and
\item $(\lambda(\alpha),\lambda(\beta))\not \in ind \implies \alpha \PO \beta \vee \beta \PO \alpha$ (any two event occurrences with dependent labels must be causally related).  \EOD
\end{enumerate}
\end{definition}

A trace in this definition is only identified unique up to  \emph{label-preserving isomorphism}. The first condition of the above definition is particularly important since two immediately causally related event occurrences will occur next to each other in at least one of its linear extensions, and thus they cannot be labeled by independent events without violating the causality relation $\prec$. This observation is the key to relate  \defref{ltraces} with quotient monoid definition of traces. 

We would like to establish a analogous definition for  comtraces. An immediate technical difficulty  is that weak causality might be cyclic, so the notion of ``immediate weak causality''  does not make sense. However, we can still deal with cycles of an so-structure by taking advantage of  the following simple fact: \emph{the weak causality relation is a strict preorder}.

\subsection{Quotient so-structure}

Let $S=(X,\PO,\sq)$ be an so-structure. We define the relation $\eqrel{\sq}\subseteq X\times X$ as
\[\alpha \eqrel{\sq} \beta \iffdf \alpha=\beta \;\vee\; \bigl(\alpha\sq \beta \wedge \beta \sq \alpha\bigr)\]

Since $\sq$ is a strict preorder, it follows that $\eqrel{\sq}$ is an equivalence relation. The relation $\eqrel{\sq}$ will be called the \emph{$\sq$-cycle equivalence relation} and an element of the quotient set $\quotient{X}{\eqrel{\sq}}$ will be called a \emph{$\sq$-cycle equivalence class}.  Define the following binary relations $\widehat{\PO}$ and $\widehat{\sq}$ on  the quotient set $\quotient{X}{\eqrel{\sq}}$ as
\begin{align}
[\alpha] \widehat{\PO} [\beta] \iffdf \alpha\not=\beta \;\wedge\; ([\alpha]\times[\beta]) \;\cap \PO\not= \emptyset \label{eq:qsos1}\\
[\alpha] \widehat{\sq} [\beta] \iffdf \alpha\not=\beta \;\wedge\; ([\alpha]\times[\beta]) \;\cap \sq\not= \emptyset  \label{eq:qsos2}
\end{align}
We call the relational structure $\left(\quotient{X}{\eqrel{\sq}},\widehat{\PO} ,\widehat{\sq} \right)$ the \emph{quotient so-structure} induced by  the so-structure $S$. 

Using this quotient construction,  we will show that every so-structure, whose weak causality relation might be cyclic, can be uniquely represented by an \emph{acyclic} quotient so-structure. 

\begin{proposition} The relational structure $\quotient{S}{\eqrel{\sq}} \df \left(\quotient{X}{\eqrel{\sq}},\widehat{\PO},\widehat{\sq}\right)$ is an so-structure, the relations $\widehat{\PO}$ and $\widehat{\sq}$ are partial orders, and for all $\alpha,\beta\in X$,
\begin{enumerate}
\item $\alpha \PO \beta\iff [\alpha] \widehat{\PO} [\beta]$ 
\item $\alpha \sq \beta \iff [\alpha] \widehat{\sq} [\beta] \vee (\alpha\not=\beta \wedge [\alpha]=[\beta])$  
\end{enumerate}
\end{proposition}
\begin{proof} Follows from \defref{sos}, and how  $\widehat{\PO}$ and $\widehat{\sq}$ are defined in  \eref{qsos1} and \eref{qsos2}. 
\qed
\end{proof}

Each $\sq$-cycle equivalence class is what Juh\'as, Lorenz and Mauser called a \emph{synchronous step} \cite{JLM08}. In their papers, they also used equivalence classes to capture synchronous steps but only for the special class of \emph{synchronous closed} so-structures, where $(\sq\setminus \PO)\cup id_{X}$ is an equivalence relation. Note that when  an so-structure is synchronous closed, the $\sq$-cycle equivalence classes of an so-structure are exactly the equivalence classes of the equivalence relation $(\sq\setminus \PO)\cup id_{X}$.

We extend their ideas by using $\sq$-cycle equivalence classes to capture ``synchronous steps'' in arbitrary so-structures. The name is justified in the following simple yet useful  proposition.

\begin{proposition} \label{prop:covlsos}
Let $S = (X,\PO,\sq)$ be an so-structure. We use $u$ and $v$  to denote some step sequences over $\PS{X}\setminus \set{\emptyset}$. Then for all $\alpha,\beta \in X$,
	\begin{enumerate}
	\item $\alpha \in [\beta] \iff  \forall \lhd \in ext(S),\; \alpha \simeq_{\lhd} \beta$
	\item There is a stratified extension $\lhd$ of $S$ such that $\Omega_{\lhd} = [\gamma_{1}]\ldots[\gamma_{k}]$. 
	\item If $[\alpha] \imm{\hat{\sq}} [\beta]$, then  there  is a stratified extension $\lhd$ of $S$ such that $\Omega_{\lhd} = [\delta_{1}]\ldots[\delta_{m}]$ and $\alpha \in [\delta_{i}]$ and $\beta \in [\delta_{i+1}]$ for some $1\le i < m$.
	\end{enumerate}
\end{proposition}

Assertion~(1) states that two elements belong to the same synchronous step of an so-structure $S$ if and only if  they must be executed simultaneously in every stratified extension of $S$. In other words, when reinterpreting each stratified  extensions of $S$ as a step sequence, all elements of a $\sq$-cycle equivalence class must always occur together within the same step. Assertion (2) says that all elements of a synchronous step must occur together as a single step in at least one stratified extension of $S$. Assertion (3) gives a sufficient condition for two synchronous steps to occur as consecutive steps in at least one stratified extension of $S$.

\begin{proof}
\textbf{1.} ($\Rightarrow$):  Since $ \alpha \in [\beta]$,  we know that $\alpha=\beta$ or $(\alpha\sq \beta\wedge \beta \sq \alpha)$. The former case is trivial. For the latter case, by \theoref{SzpStrat},  we have
$\forall \lhd \in ext(S).\; \alpha \lhd^{\frown} \beta$ and $\forall \lhd \in ext(S).\; \beta \lhd^{\frown} \alpha$. (Recall that $ \lhd^{\frown}  := \lhd \;\cup  \frown_{\lhd}$.) Thus, it follows that  $\forall \lhd \in ext(S).\; \alpha \frown_{\lhd} \beta$.

($\Leftarrow$): By definition, we have  $\simeq_{\lhd}:= \frown_{\lhd} \cup\; id_{X}$, so we consider two cases. The first case when $\alpha=\beta$ is trivial. The second case when  we have $\forall \lhd \in ext(S).\; \alpha \frown_{\lhd} \beta$ and $\alpha\not=\beta$. Thus, by \theoref{SzpStrat}, $\alpha \sq \beta$ and $\beta \sq \alpha$, which implies  $\alpha$ and $\beta$ belong to the same equivalence class.

\textbf{2. }We will construct a step sequence $s = [\gamma_{1}]\ldots[\gamma_{k}]$ which can be converted back to a stratified extension $\lhd$ of the so-structure $S$ as follows. 
Since $P=\left(\quotient{X}{\eqrel{\sq}},\widehat{\sq} \right)$ is a poset, we can simply choose $s$ to be an arbitrary total order extension of $P$. From this, it is not hard to check that the stratified order $(X,\lhd)$ representing the step sequence $s$ is an extension of $S$.

\textbf{3. }Similarly to \textbf{2.}, we can choose the step sequence $t = [\delta_{1}]\ldots[\delta_{m}]$ to be a total order extension of the poset $P$. Moreover, since $[\alpha]$ and $[\beta]$ are immediately related by the partial order $\widehat{\sq}$, we can easily choose the step sequence $t$ such that $[\alpha]$ and $[\beta]$ occur consecutively on $t$. 
\qed
\end{proof}

\subsection{Using quotient so-structure to define comtrace}
We need to define  \textit{label-preserving isomorphisms} for labeled so-structures more formally. A tuple $T = (X,P,Q,\lambda)$ is a \textit{labeled relational structure} if and only if $(X,P,Q)$ is a relational structure and $\lambda$ is a function with domain $X$. If $(X,P,Q)$ is an so-structure, then $T$ is a \textit{labeled so-structure}.

\begin{definition}[label-preserving isomorphism] Let $T_{1}$ and $T_{2}$ be labeled relational structures, where  we let $T_{i}=(X_{i},P_{i},Q_{i},\lambda_{i})$. We write $T_{1}\cong T_{2}$ if and only if $T_1$ and $T_2$ are  \emph{label-preserving isomorphic} (\emph{lp-isomorphic}). In other words, there is a bijection $f: X_{1} \rightarrow X_{2}$ such that for all $\alpha,\beta \in  X_{1}$, 
\begin{enumerate}
\item $(\alpha,\beta) \in P_{1} \iff (f(\alpha),f(\beta))\in P_{2}$
\item $(\alpha,\beta) \in Q_{1} \iff (f(\alpha),f(\beta))\in Q_{2}$
\item $\lambda_{1}(\alpha) = \lambda_{2}(f(\alpha))$
\end{enumerate}
Such function $f$ is called a \emph{label-preserving isomorphism} (\emph{lp-isomorphism}).  \EOD
\end{definition}

Note that all notations, definitions and results for so-structures are applicable to labeled so-structures. We also write $[T]$ or $\seq{X,P,Q,\lambda}$ to denote the lp-isomorphic class of a labeled relational structure $T=(X,P,Q,\lambda)$. We  will not distinguish  an lp-isomorphic class $\seq{T}$ with a single labeled relational structure $T$ when it does not cause ambiguity.

We are now ready to give an alternative definition for comtraces. To avoid confusion with the comtrace notion by Janicki and Koutny in \cite{JK95}, we will use the term \emph{lsos-comtrace} to denote a comtrace defined using our definition. 

\begin{definition}[lsos-comtrace]  \label{def:lcomtrace}
A \emph{lsos-comtrace} over a comtrace alphabet $\theta=(E,sim,ser)$ is (an lp-isomorphic class of) a finite labeled so-structure $\seq{X,\PO,\sq,\lambda}$ such that $\lambda:X\rightarrow E$  and for all elements $\alpha\not=\beta$ of $ X$,
\begin{center}
\tikzstyle{mybox} = [draw=black, very thick,
    rectangle, rounded corners, inner sep=5pt, inner ysep=10pt]

\begin{tikzpicture}
\node [mybox] (box){%
    \begin{minipage}{0.6\textwidth}
\begin{enumerate}
\item[] \textsf{LC1:\mbox{\hspace{5mm}}} $[\alpha] (\imm{\hat{\sq}}\cap \hat{\PO}) [\beta] \implies \lambda\bigl[[\alpha]\bigr]\times \lambda\bigl[[\beta]\bigr] \nsubseteq ser$
\item[] \textsf{LC2:\mbox{\hspace{5mm}}} $[\alpha] (\imm{\hat{\sq}}\setminus \hat{\PO}) [\beta] \implies \lambda\bigl[[\beta]\bigr]\times \lambda \bigl[[\alpha]\bigr] \nsubseteq ser$
\item[] \textsf{LC3:\mbox{\hspace{5mm}}} $\forall A,B \in \PS{[\alpha]} \setminus \set{\emptyset},\; A\cup B = [\alpha] \implies \lambda[A]\times \lambda[B] \not \subseteq ser$
\item[] \textsf{LC4:\mbox{\hspace{5mm}}} $(\lambda(\alpha),\lambda(\beta))\not \in ser \implies \alpha \PO \beta \vee \beta \sq \alpha$
\item[] \textsf{LC5:\mbox{\hspace{5mm}}} $(\lambda(\alpha),\lambda(\beta))\not\in sim \implies \alpha\PO \beta \vee \beta \PO\alpha$ 
\end{enumerate}
    \end{minipage}
};
\end{tikzpicture}%
\end{center}
We write $\LCT(\theta)$ to denote the class of all lsos-comtraces over $\theta$.  \EOD
\end{definition}

To understand the first three conditions \textsf{LC1}--\textsf{LC3} of this definition, it is more intuitive to consider the Hasse diagram  $\left(\quotient{X}{\eqrel{\sq}},\imm{\widehat{\sq}} \right)$ of the poset $\left(\quotient{X}{\eqrel{\sq}},\widehat{\sq} \right)$. Each vertex $[\alpha]$ of this Hasse diagram is a $\sq$-cycle equivalence class, so all  events in $[\alpha]$ must happen simultaneously, and thus each equivalence class $[\alpha]$ can be seen as a ``composite event''. Hence, condition \textsf{LC3} imposes  that we cannot serialize $[\alpha]$ into two separate steps. Also, given $[\alpha] \imm{\widehat{\sq}} [\beta]$ on this Hasse diagram, we know that the ``composite event'' $[\alpha]$ must happen not later than the ``composite event'' $[\beta]$, and  there are two possible cases. If $[\alpha]$ must happen ealier than $[\beta]$, then  condition \textsf{LC1} ensures that the events in $[\alpha]$ and $[\beta]$ cannot be put back together into a single step, i.e., $\lambda\bigl[[\alpha]\bigr]\times \lambda\bigl[[\beta]\bigr] \nsubseteq ser$. Condition \textsf{LC2} is the dual of condition \textsf{LC1} and takes care of the remaining case when $[\alpha]$ does not have to happen ealier than $[\beta]$. 

The last two conditions are needed to ensure that the relationships $ser$ and $sim$ are properly translated into the ``ealier than'' relation $\PO$ and the ``not later than'' relation $\sq$ of the so-structure. It is important to note that the definition is still valid if  we substitute \textsf{LC4} and \textsf{LC5} with the following two conditions:
\begin{center}
\tikzstyle{mybox} = [draw=black, very thick,
    rectangle, rounded corners, inner sep=5pt, inner ysep=10pt]

\begin{tikzpicture}
\node [mybox] (box){%
    \begin{minipage}{0.6\textwidth}
\begin{itemize}
\item[] \textsf{LC4':\mbox{\hspace{5mm}}} $\lambda\bigl[[\alpha]\bigr] \times \lambda\bigl[[\beta]\bigr]\not \subseteq ser \wedge [\alpha]\not=[\beta] \;\implies [\alpha] \widehat{\PO} [\beta] \vee [\beta] \widehat{\sq} [\alpha]$
\item[] \textsf{LC5':\mbox{\hspace{5mm}}} $\lambda\bigl[[\alpha]\bigr] \times \lambda\bigl[[\beta]\bigr]\not\subseteq sim  \wedge [\alpha]\not=[\beta] \implies [\alpha] \widehat{\PO} [\beta] \vee [\beta] \widehat{\PO} [\alpha]$ 
\end{itemize}
    \end{minipage}
};
\end{tikzpicture}%
\end{center}
Then all conditions will involve only the $\sq$-cycle equivalence classes.  In other words, this definition can be seen as a characterization of comtraces using the quotient so-structure $\left(\quotient{X}{\eqrel{\sq}},\widehat{\PO} ,\widehat{\sq} \right)$.

\begin{example} \label{ex:comtrace2}
Let $\theta=(E,sim,ser)$ be a comtrace alphabet, where 
\begin{align*}
E&=\set{a,b,c}\\
sim &= \bset{(a,b),(b,a),(a,c),(c,a),(b,c),(c,b)}\\
ser&=\bset{(a,b),(b,a),(a,c)}
\end{align*}
The set of all possible steps is
$\E=\set{\{a\},\{b\},\{c\},\set{b,c},\set{a,c},\set{a,b},\set{a,b,c}}$. The lp-isomorphic class  of the labeled so-structure $T=(X,\PO,\sq,\lambda)$ shown  in \figref{f1} is an lsos-comtrace. The dashed edges denote the $\sq$ relation and the solid edges denote the $\PO$ relation. The graph in \figref{f2} represents the labeled quotient  so-structure $\quotient{T}{\eqrel{\sq}}=(\quotient{X}{\eqrel{\sq}},\widehat{\PO},\widehat{\sq},{\lambda'})$ of $T$, where we define ${\lambda'}(A) := \lambda[A]$.
\begin{figure}[ht]
\begin{minipage}{0.32\textwidth}\centering
\begin{tikzpicture}[node distance=11mm,->,>=stealth']
        \tikzstyle{every node}=[circle,fill=black!30,minimum size=15pt,inner sep=0pt]
        \node (A) {$a$};
        \node (C) [below of=A,xshift=1.5cm] {$c$};
	\node (B) [above of=C,xshift=1.5cm] {$c$};
        \node (D) [below of=C,xshift=-1.5cm] {$b$};
        \node (E) [below of=C,xshift=1.5cm] {$b$};

        \path [thick,shorten >=1pt] 
        			(A) edge[bend left] (B) edge[bend right] (E)
                         (C) edge (B) edge (E)
                         (D) edge (C) edge[bend right] (E) edge[bend left] (B);

        \path [dashed,thick,shorten >=1pt]
        			(A) edge (C)
			(B) edge[bend left] (E)
			(E) edge[bend left] (B);
\end{tikzpicture}			
\caption{lsos-comtrace $[T]$}
\label{fig:f1}
\end{minipage}
\begin{minipage}{0.25\textwidth}\centering
\vspace{-2cm}
\begin{tikzpicture}
  \draw [black,-to,thick,snake=snake,segment amplitude=.4mm,
         segment length=2mm,line after snake=1mm]
    (0,0) -- (2.5,0)
    node [above=1mm,midway,text width=3cm,text centered]
      { quotient \\ construction};
\end{tikzpicture}
\end{minipage}
\begin{minipage}{0.37\textwidth}\centering
\begin{tikzpicture}[node distance=14mm]
        \tikzstyle{every node}= [rectangle,fill=black!30,minimum size=20pt,inner sep=0pt]

        \node (A') {$\set{a}$};
        \node (C') [below of=A',xshift=1.5cm] {$\set{c}$};
	\node (B') [right of=C',xshift=.7cm] {$\set{b,c}$};
        \node (D') [below of=C',xshift=-1.5cm] {$\set{b}$};

        \path [thick,shorten >=1pt,-stealth'] 
        			(A') edge[bend left] (B')  
			 (C') edge (B') 
                         (D') edge (C') edge[bend right] (B');

        \path [dashed,thick,shorten >=1pt,-stealth']
        			(A') edge (C');
\end{tikzpicture}
\caption{the quotient structure $\quotient{T}{\eqrel{\sq}}$}
\label{fig:f2}
\end{minipage}
\end{figure}
It is important to note that the quotient construction collapses the two rightmost nodes of the graph in \figref{f1} since they belong to the same $\sq$-cycle equivalence class. The result is a much simpler and more compact representation in \figref{f2}.

The lsos-comtrace $[T]$ actually corresponds to the comtrace 
\[[\set{a,b}\set{c}\set{b,c}] = \bset{\set{a,b}\set{c}\set{b,c},\set{a}\set{b}\set{c}\set{b,c},\set{b}\set{a}\set{c}\set{b,c},\set{b}\set{a,c}\set{b,c}},\] since the step sequences in this comtrace when reinterpreted as stratified orders are exactly the two stratified extensions of the lsos-comtrace $T$. We will show this relationship formally in \secref{ct2lct}. 
\EOD 
\end{example}

\begin{remark}
\defref{lcomtrace} can be extended to define \emph{infinite comtrace} as follows. Instead of asking $X$ to be finite, we require a labeled so-structure to be  \emph{initially finite} (cf. \cite{JK97}), i.e., $\bset{\alpha\in X\mid \alpha \sq\beta}$ is finite for all $\beta\in X$. The initial finiteness gives us a sensible interpretation that before any event there can only be finitely many events. \EOD
\end{remark}

\subsection{Canonical representation of lsos-comtrace \label{sec:can}}
Since each lsos-comtrace is defined as a class of lp-isomorphic labeled so-structures, it might seem tricky to work  with lsos-comtraces. Fortunately, the ``no autoconcurrency'' property, i.e., the relations $sim$ and $ser$  are irreflexive, gives us a \emph{canonical} way to enumerate the events of an lsos-comtrace  similar to how the events of a comtrace are enumerated.

Given an lsos-comtrace $T=\seq{X,\PO,\sq,\lambda}$  over a comtrace alphabet $\theta=(E,sim,ser)$,  a stratified order $\lhd\in ext(T)$ can be seen as a step sequence $\Omega_{\lhd}=A_{1}\ldots A_{k}$ satisfying the following properties.

\begin{proposition}\label{prop:validss}  $\;$
\begin{enumerate}
\item For all $1\le i \le k$, $|A_{i}|=|\lambda[A_{i}]|$
\item $\map(\lambda,\Omega_{\lhd}) = \lambda[A_{1}]\ldots \lambda[A_{k}]$ is a step sequence over the comtrace alphabet $\theta$. 
\end{enumerate}

\end{proposition}
This proposition ensures that  $u=\map(\lambda,\Omega_{\lhd})$ is a valid step sequence over the comtrace alphabet $\theta$.
\begin{proof}\textbf{1. } Intuitively, this follows from the fact that $sim$ is irreflexive, which guarantees that different occurrences of the same events cannot occurs simultaneously in any stratified extension of $T$. More formally, assume $\alpha,\beta \in A_i$ and $\alpha\not=\beta$. Thus, $\alpha\frown_{\lhd}\beta$.  By \corref{SzpStrat} (2), this implies that  $\alpha \frown_\PO \beta$. Hence, by \textsf{LC5} of \defref{lcomtrace},  $(\lambda(\alpha),\lambda(\beta))\in sim$. Since $sim$ is irreflexive, this implies that any two distinct $\alpha$ and $\beta$ in $A_i$ have different labels. Thus, $|A_{i}|=|\lambda[A_{i}]|$ for all $1\le i \le k$.

\textbf{2. }From the proof of \textbf{1.}, we know that  for any two distinct $\alpha,\beta \in A_i$ we have $(\lambda(\alpha),\lambda(\beta))\in sim$. Thus, $\lambda[A_{i}]\in \E_{\theta}$ for all $1\le i \le k$.
\qed
\end{proof}

Given an lsos-comtrace $T=\seq{X,\PO,\sq,\lambda}$ as above, we define a  function $\enu$ with domain $X$ as 
\[\enu(x)  = \lambda(x)^{(i)},\]
where the index $i := |\set{y \in X \mid y\PO x \;\wedge\; \lambda(y)=\lambda(x)}|+1$. In other words, $i-1$  is the number of elements that occur earlier than $x$ and have the same labels as $x$. Thus, the index $i=1$ when $x$ is the first occurrence of the event $\lambda(x)$ with respect to the relation $\PO$, and in general $\enu(x)  = \lambda(x)^{(i)}$ if $x$ is the $i^{\rm th}$ occurrence of the event $\lambda(x)$ with respect to  $\PO$.  Note that we can enumerate events with the same label in  this way because it follows from \textsf{LC5} of \defref{lcomtrace} that events with the same label are totally ordered by $\PO$.

Fix a stratified extension  $\lhd$ of the above lsos-comtrace $T$, and assume that  $\Omega_{\lhd}=A_{1}\ldots A_{k}$. Then by \propref{validss}, we know that the step sequence $u=\map(\lambda,\Omega_{\lhd})$ is a valid step sequence over the comtrace alphabet $\theta$.  Recall that $\h{u}=\h{A_{1}}\ldots \h{A_{k}}$ denotes the enumerated step sequence of $u$ and $\Sigma_{u}$ denotes the set of event occurrences in $u$.  Then it is not hard to see that the range of the function $\enu$ defined above is exactly the set $\Sigma_{u}$, i.e., $\enu[X] = \Sigma_{u}$, and this holds regardless of the choice of the stratified extension $\lhd \in ext(T)$.

 We define the \emph{enumerated so-structure} of $T$ to be the labeled so-structure $T_{0}=(\Sigma,\PO_{0},\sq_{0},\ell)$, where we let $\Sigma = \enu[X]$, the function $\ell:\Sigma \rightarrow E$ is as defined as in \secref{steps}, and the two relations $\PO_{0},\sq_{0}\;\subseteq \Sigma \times \Sigma$ are defined as:
\begin{align*}
\alpha \PO_{0} \beta &\iffdf \enu^{-1}(\alpha) \PO \enu^{-1}(\beta)\\
\alpha \sq_{0} \beta &\iffdf \enu^{-1}(\alpha) \sq \enu^{-1}(\beta)
\end{align*}
Here $\enu^{-1}$ denotes the inverse of the function $\enu$. The inverse of $\enu$ is well-defined since it is not hard to see that the function $\enu:X\rightarrow \Sigma$ is bijective.

Clearly the  enumerated so-structure $T_0$ can be uniquely determined from $T$ in this definition.  Since $T_0$ is constructed from $T$ by renaming the elements in $X$ using the function $\enu$, we can easily show the following relationships:

\begin{proposition}\label{prop:isoExt} $\;$ 
\begin{enumerate}
\item $T$ and $T_{0}$ are  lp-isomorphic under the mapping $\enu$. 
\item For every stratified extension  $\lhd$ of $T$, if $u=\map(\lambda,\Omega_{\lhd})$, then  $(X,\lhd,\lhd^{\frown},\lambda)$ and $(\Sigma,\lhd_u,\lhd_u^{\frown},\ell)$  are lp-isomorphic under the mapping $\enu$, and we also have $\lhd_u\in ext(T_{0})$.   
\end{enumerate} 
\end{proposition}

In other words, the mapping $\enu:X\rightarrow \Sigma$ plays 
the role of both the lp-isomorphism from $T$ to $T_{0}$ and the lp-isomorphism from the stratified extension $(X,\lhd,\lambda)$ of $T$  to the stratified extension   $(\Sigma,\lhd_{u},\ell)$ of $T_{0}$. These relationships can be best captured using the following commutative diagram.
\begin{center}
\begin{tikzpicture}[description/.style={fill=white,inner sep=2pt},>=stealth']
\matrix (m) [matrix of math nodes, row sep=4em,
column sep=2.5em, text height=2ex, text depth=0.3ex]
{  (X,\PO,\sq,\lambda) & & (\Sigma,\PO_{0},\sq_{0},\ell) \\
 (X,\lhd,\lhd^{\frown},\lambda) &  & (\Sigma,\lhd_{u},\lhd_{u}^{\frown},\ell) \\ };
\path[->,font=\small,thick]
	(m-1-1) edge node[auto] {$ \enu $} (m-1-3)
	(m-2-1) edge node[auto] {$ \enu $} (m-2-3);
\path[right hook->,font=\small,thick]
	(m-1-1) edge node[auto,swap] {$ id_{X} $} (m-2-1)
	(m-1-3) edge node[auto] {$ id_{\Sigma} $} (m-2-3);
	
\end{tikzpicture}
\end{center}

We can also observe that two lsos-comtraces are identical if and only if they define the same enumerated so-structure. Henceforth, we define the  \emph{canonical representation} of $T$ to be the enumerated so-structure of $T$. \bigskip\\
\begin{minipage}{0.6 \textwidth}
\begin{example} The canonical presentation $T_{0} = (\Sigma, \PO_{0}, \sq_{0},\ell)$ of the lsos-comtrace $T$ from \exref{comtrace2} is given in the figure on the right. The solid edges of the graph represents the relation $\PO_{0}$ and the dashed edges represent the relation $\sq_{0}$. The nodes of this graph are labeled with the elements of $\Sigma$. The labeling function $\ell$ returns the label of each element of $\Sigma$ by simply  getting rid of its superscript. For example, $\ell(a^{(1)}) = a$ and $\ell(b^{(2)}) = b$. \EOD
\end{example}
\end{minipage}
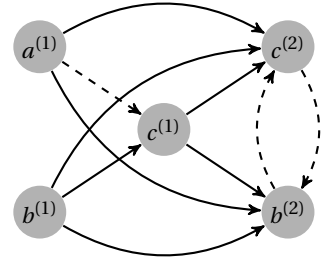
\begin{minipage}{0.4 \textwidth}
\begin{center}
\begin{footnotesize}
\begin{tikzpicture}[scale=1.1,->,>=stealth']
        \tikzstyle{every node}=[circle,fill=black!30,minimum size=20pt,inner sep=0pt]

        \node (A) at (0,0) {$a^{(1)}$};
        \node (C) at (1.5,-1) {$c^{(1)}$};
	\node (B) at (3,0) {$c^{(2)}$};
        \node (D) at (0,-2) {$b^{(1)}$};
        \node (E) at (3,-2) {$b^{(2)}$};

        \path [thick] 
        			(A) edge[bend left] (B) edge[bend right] (E)
                         (C) edge (B) edge (E)
                         (D) edge (C) edge[bend right] (E) edge[bend left] (B);

        \path [dashed,thick]
        			(A) edge (C)
			(B) edge[bend left] (E)
			(E) edge[bend left] (B);
\end{tikzpicture}	
\end{footnotesize}
\end{center}
\end{minipage}

\section{Representation theorems \label{sec:representation}}
This section contains the main technical contribution of this paper by showing that for a given comtrace alphabet $\theta$, the comtrace monoid $\E^*_{\theta}/\!\!\equiv_{\theta}$, the set of lsos-comtraces $\LCT(\theta)$ and the set of cd-graphs $\DCD(\theta)$ are three equivalent ways of talking about the same class of objects. 

\subsection{Representation theorem for comtraces and lsos-comtraces \label{sec:ct2lct}}
The goal of this section is to prove the first representation theorem which establishes the representation mappings between $\E^*/\!\!\equiv_{\theta}$ and $\LCT(\theta)$. Before doing so, we need some preliminary results.

\begin{proposition} \label{prop:stratsubset}
Let $S=(X,\PO,\sq)$ and $S'=(X,\PO',\sq')$ be stratified order structures such that $ext(S)\subseteq ext(S')$. Then $S'\subseteq S$.
\end{proposition}
\begin{proof} Follows from \theoref{SzpStrat}.  
\qed
\end{proof}

For the next two lemmas, we let $T$ be an lsos-comtrace over a comtrace alphabet $\theta=(E,sim,ser)$. Let $T_{0}=(\Sigma,\PO_{0},\sq_{0},\ell)$ be the canonical representation of $T$. Let $\lhd_{0}\in ext(T_0)$ and $u=\map(\ell,\Omega_{\lhd_{0}})$. Since $u$ is a valid step sequence in $\E^*$ (by \propref{validss}), we can construct the so-structure $S_{[u]} = (\Sigma_{u},\PO_{[u]},\sq_{[u]})$ as in  \defref{s2sos}, where we have seen in \secref{can} that   $\Sigma =\Sigma_{u}$. Our goal is to show that the so-structure $S_{[u]}$  is exactly the canonical representation $T_{0}$ of $T$.

\begin{lemma} $S_{[u]} \subseteq (\Sigma_{u},\PO_{0},\sq_{0})$.
\label{lem:l1}
\end{lemma}
\begin{proof} By \propref{so-cl}, to show $S_{[u]} = (\Sigma_{u},\PO_{u},\sq_{u})^{\lozenge} \subseteq (\Sigma_{u},\PO_{0},\sq_{0})$, it suffices to show $(\Sigma_{u},\PO_{u},\sq_{u}) \subseteq (\Sigma_{u},\PO_{0},\sq_{0})$. Intuitively, the lemma holds since $u$ is the step sequence representation of the stratified extension $\lhd$ of $T_{0}$, and so the relations $\PO_{u}$ and $\sq_{u})$ defined from $u$ cannot `violate' (i.e. must be contained in) the relations $\PO_{0}$ and $\sq_{0}$ respectively. The formal proof goes as follows.

($\PO_{u}\subseteq \PO_{0}$): Assume $\alpha\PO_{u}\beta$. Then from \defref{s2inv}, $\alpha\lhd_u\beta$ and  $(\ell(\alpha),\ell(\beta))\notin ser$. Since $(\ell(\alpha),\ell(\beta))\notin ser$, it follows from \defref{lcomtrace} that  $\alpha \PO_{0} \beta$ or $\beta\sq_{0} \alpha$.  Suppose for a contradiction that $\beta\sq_{0} \alpha$, then by \theoref{SzpStrat},  $\forall \lhd \in ext(T_{0}),\; \beta \lhd^{\frown}\alpha$. Since $T_0$ is the canonical representation of $T$, it is important to observe that $\lhd_0=\lhd_u$. But since we assumed that $\lhd_{0}\in ext(T_0)$, it follows that $\lhd_{u}\in ext(T_{0})$ and $\alpha \lhd_{u} \beta$, which contradicts that $\beta\sq_{0} \alpha$. Hence, $\alpha \PO_{0} \beta$.

($\sq_{u}\subseteq \sq_{0}$):  Can be shown similarly.
\qed
\end{proof}

\begin{lemma} $S_{[u]} \supseteq (\Sigma_{u},\PO_{0},\sq_{0})$.
\label{lem:l2}
\end{lemma}

In this proof, we will include subscripts for equivalence classes to avoid confusing the elements from quotient set $\quotient{\Sigma_{u}}{\equiv_{\sq_{0}}}$ with the elements from the quotient comtrace monoid $\quotient{\E^*_{\theta}}{\equiv_{\theta}}$. In other words, we write $[\alpha]_{\equiv_{\sq_{0}}}$ to denote an element of the quotient set $\quotient{\Sigma_{u}}{\equiv_{\sq_{0}}}$, and we write $[u]_{\theta}$ to denote the comtrace generated by the step sequence $u$.

\begin{proof}Let $S'=(\Sigma_{u},\PO_{0},\sq_{0})$. To show $S_{[u]} \supseteq S'$, by \propref{stratsubset}, it suffices to show $ext(S_{[u]} )\subseteq ext(S')$. From \theoref{com2sos}, we know that $ext(S_{[u]_{\theta}}) = \set{\lhd_w\mid w\in [u]_{\theta}}$. Thus we only need to show that $\lhd_w \in ext(S')$ for all $w\in [u]_{\theta}$. The main idea of the proof is simple: since every step sequence in $[u]_{\theta}$ is generated from the sequence $u$ according to how $sim$ and $ser$ are defined and $u$ is the step sequence representation of the stratified extension $\lhd$ of $S'$, it follows from how lsos-comtraces are defined that every stratified order $\lhd_{s}$, where $s\in [u]$, must also be a stratified extension of $S'$. Unfortunately, the actual proof is a bit technical and goes as follows.

We observe that from $u$, by \defref{commonoid}, we can generate all the step sequences in the comtrace $[u]_{\theta}$ in stages using the following recursive definition:
\begin{align*}
D^0(u)&\df \set{u}\\
D^n(u)&\df \set{w\mid w\in D^{n-1}(u)\;\vee\; \exists v\in D^{n-1}(u),\ (\ v\eqa_{\theta} w\; \vee\; v\eqa_{\theta}^{-1} w)} 
\end{align*}
Since the set $[u]_{\theta}$ is finite, $[u]_{\theta}=D^n(u)$ for some stage $n\ge 0$. The proof is complete if we can show the following claim.
\begin{center}
\tikzstyle{mybox} = [draw=black, very thick,
    rectangle, rounded corners, inner sep=5pt, inner ysep=5pt]

\begin{tikzpicture}
\node [mybox] (box){%
    \begin{minipage}{0.65\textwidth}\centering
        \textbf{Claim 1.} For all $n\in \mathbb{N}$, we have $\lhd_w\in ext(S')$ for all $w\in D^n(u)$.
    \end{minipage}
};
\end{tikzpicture}%
\end{center}
We prove Claim 1 by induction on $n$.\\
\textbf{Base case:} When $n=0$, $D^0(u)=\set{u}$. Since $\lhd_0 \in ext(T)$, it follows from \propref{isoExt} that $\lhd_u \in ext(S')$. \\
\textbf{Inductive case:} When $n>0$, let $w$ be an element of $D^n(u)$. Then either $w\in D^{n-1}(u)$ or $w\in (D^n(u)\setminus D^{n-1}(u))$. For the former case, by inductive hypothesis, $\lhd_w\in ext(S')$. For the latter case, there must be some element $v\in D^{n-1}(u)$ such that $v\eqa_{\theta} w$ or  $v\eqa_{\theta}^{-1} w$. By induction hypothesis, we already known $\lhd_v\in ext(S')$. We want to show that $\lhd_w \in ext(S')$. There are two cases to consider:\smallskip\\
\textit{\textbf{Case (i):}} \\
When $v\eqa_{\theta} w$, by \defref{commonoid}, there are some $y,z\in \E_{\theta}^*$ and steps $A,B,C\in \E_{\theta}$ such that $v=yAz$ and $w=yBCz$ where $A$, $B$, $C$ satisfy $B\cap C =\emptyset$ and $B\cup C = A$ and $B\times C\subseteq ser$. Let $\h{v}=\h{y}\h{A}\h{z}$ and $\h{w}=\h{y}\h{B}\,\h{C}\h{z}$ be enumerated step sequences of $v$ and $w$ respectively. 

Suppose for a contradiction that $\lhd_w \not\in ext(S')$. By \defref{extsos}, there are $\alpha\in \h{C}$  and $\beta \in \h{B}$ such that $\alpha \sq_{0} \beta$. We now consider the quotient set $\quotient{\h{A}}{\equiv_{\sq_{0}}}$. By  \propref{covlsos} (1), $\quotient{\h{A}}{\equiv_{\sq_{0}}} \subseteq \quotient{\Sigma_{u}}{\equiv_{\sq_{0}}}$. 
Since $\alpha \sq_{0} \beta$, it follows that $[\alpha]_{\equiv_{\sq_{0}}}\hat{\sq}_{0} [\beta]_{\equiv_{\sq_{0}}}$. Thus, from the fact that $\hat{\sq}_{0}$ is a partial order, there must exists a chain
\begin{align}
[\alpha]_{\equiv_{\sq_{0}}}=\; [\gamma_1]_{\equiv_{\sq_{0}}} \;\imm{\hat{\sq}_{0}}\; [\gamma_2]_{\equiv_{\sq_{0}}}\;\imm{\hat{\sq}_{0}}\; \ldots\;\imm{\hat{\sq}_{0}}\;[\gamma_k]_{\equiv_{\sq_{0}}}\;= [\beta]_{\equiv_{\sq_{0}}} 
\label{eq:chain1}
\end{align}
We want to show the following claim.
\begin{center}
\tikzstyle{mybox} = [draw=black, very thick,
    rectangle, rounded corners, inner sep=5pt, inner ysep=10pt]

\begin{tikzpicture}
\node [mybox] (box){%
    \begin{minipage}{0.83\textwidth}\centering
        \textbf{Claim 2.} There are two consecutive elements $[\gamma_i]_{\equiv_{\sq_{0}}}$ and 
        	$[\gamma_{i+1}]_{\equiv_{\sq_{0}}}$ on the chain such that 
	\begin{align*}
	\text{(a) \quad}\ell \Bigl[[\gamma_{i+1}]_{\equiv_{\sq_{0}}}\Bigr]\times \ell \Bigl[[\gamma_i]_{\equiv_{\sq_{0}}}\Bigr]\subseteq ser	
	&&\text{(b) \quad }\neg \bigl([\gamma_i]_{\equiv_{\sq_{0}}} \hat{\PO}_{0} [\gamma_{i+1}]_{\equiv_{\sq_	{0}}}\bigr)
	\end{align*}
    \end{minipage}
};
\end{tikzpicture}%
\end{center}
Note that Claim 2 gives us the desired contradiction for this case by violating the assumption that $T_{0}$ satisfies condition \textsf{LC2} of \defref{lcomtrace}. We can show Claim 2 as follows.

Our proof of Claim~2~(a) is an instance of the following simple combinatorial observation. 
\begin{quote}
Assume that  the elements of a chain are colored such that each element can only be colored either black or white, and furthermore we know that the first element of the chain is colored black and the last element of the chain is colored white. Then there must be some point on the chain where the color is switched from black to white (i.e. there must be two consecutive elements on the chain that are colored black and white respectively).
\end{quote}
We apply this observation to show Claim~2~(a)  as follows. By \theoref{SzpStrat} and the fact that $\lhd_v\in ext(S')$, we know that $\gamma_i \in \h{A}$ for all $i$. This follows from the fact that  the chain \eref{chain1} implies that every $\gamma_i$ must always occur between $\alpha$ and $\beta$ in all stratified extensions of $S'$ and $\alpha,\beta \in \h{A}$. Hence, by \propref{covlsos} (1), we have $[\gamma_i]_{\equiv_{\sq_{0}}}\subseteq \h{A}$ for all $i$, $1\le i \le k$. Also from \textsf{LC3} of \defref{lcomtrace} and that $B\times C \subseteq ser$, we know that either $[\gamma_i]_{\equiv_{\sq_{0}}}\subseteq \h{B}$ or $[\gamma_i]_{\equiv_{\sq_{0}}}\subseteq \h{C}$ for all $i$, $1\le i \le k$. Now we note  that the first element on the chain $[\gamma_1]_{\equiv_{\sq_{0}}}=[\alpha]_{\equiv_{\sq_{0}}}\subseteq \h{C}$ and the last element on the chain $[\gamma_k]_{\equiv_{\sq_{0}}}=[\beta]_{\equiv_{\sq_{0}}}\subseteq \h{B}$. Thus, there exist two consecutive elements $[\gamma_i]_{\equiv_{\sq_{0}}}$ and $[\gamma_{i+1}]_{\equiv_{\sq_{0}}}$ on the chain such that $[\gamma_i]_{\equiv_{\sq_{0}}}\subseteq \h{C}$ and $[\gamma_{i+1}]_{\equiv_{\sq_{0}}}\subseteq \h{B}$.  Since $B\times C \subseteq ser$, we have just shown Claim~2~(a).   

Since $\lhd_v\in ext(S')$ and $\gamma_i\frown_{\lhd_v}\gamma_{i+1}$, by \theoref{SzpStrat}, Claim~2~(b) also follows.\bigskip\\
\textit{\textbf{Case (ii):}} \\
When $v\eqa_{\theta}^{-1} w$, by \defref{commonoid}, there are some $y,z\in \E_{\theta}^*$ and steps $A,B,C\in \E_{\theta}$ such that $v=yBCz$ and $w=yAz$  where $A$, $B$, $C$ satisfy $B\cap C =\emptyset$ and $B\cup C = A$ and $B\times C\subseteq ser$. Using an argument dual to the previous case, we can show that this leads to a contradiction with \textsf{LC1} of \defref{lcomtrace}. 
\qed
\end{proof}

We also need to show that the labeled so-structure defined from each comtrace is indeed an lsos-comtrace. In other words, we need to show the following lemma.

\begin{lemma}  \label{lem:l3}
Let $\theta=(E,sim,ser)$ be a comtrace alphabet. Given a step sequence  $u \in \E_{\theta}^{*}$, the lp-isomorphic class $\seq{\Sigma_{[u]},\PO_{[u]},\sq_{[u]},\ell}$ is an lsos-comtrace over $\theta$. 
\end{lemma}
\begin{proof} Let $T=\bigl(\Sigma_{[u]},\PO_{[u]},\sq_{[u]},\ell \bigr)$. From \theoref{com2sos}, $T$ is a labeled so-structure. It only remains to show that $T$ satisfies conditions \textsf{LC1}-\textsf{LC5} of \defref{lcomtrace}. 

\textsf{LC1}: Assume  $[\alpha] (\imm{\hat{\sq}}\cap \hat{\PO}) [\beta]$ and suppose for a contradiction that $\ell\bigl[[\alpha]\bigr]\times \ell\bigl[[\beta]\bigr] \subseteq ser$. Then from \propref{covlsos} (3), there exists $\lhd \in ext(T)$ such that $\Omega_{\lhd} = v[\alpha][\beta]w$. By \theoref{com2sos}, we have $\lhd \in  \bset{\lhd_{s}\mid s \in [u]} = ext(S_{[u]})$. Thus, the step sequence $\map(\ell,v[\alpha][\beta]w)$ is in the comtrace $[u]$. Since $\ell\bigl[[\alpha]\bigr]\times \ell\bigl[[\beta]\bigr]\subseteq ser$, the step sequence $\map(\ell,u([\alpha]\cup [\beta])v)$ is also in  $[u]$. Hence, the stratified order $\lhd'$ satisfying $\Omega_{\lhd'} =  u([\alpha]\cup [\beta])v$ is also an extension of $T$, where $\alpha \frown_{\lhd'} \beta$. But this contradicts that $[\alpha] \hat{\PO} [\beta]$. We can also show \textsf{LC2} and \textsf{LC3} similarly. 

\textsf{LC4}: Follows from Definitions \ref{def:s2inv} and \ref{def:s2sos} and the $\lozenge$-closure definition.

\textsf{LC5}: Assume that $\alpha \not \PO \beta$ and $\beta \not \PO \alpha$, then $\alpha \frown_{\PO}\beta$. Thus,  it follows from \corref{SzpStrat} that there exists $\lhd \in ext(T)$ where $\alpha \frown_\lhd \beta$. Since $\bset{\lhd_{s}\mid s \in [u]} = ext(S_{[u]})$, there exists a step sequence $s\in [u]$ such that $s = \map(\ell,\Omega_\lhd)$. This implies $\alpha$ and $\beta$ belong to the same step in the step sequence  $\h{s}$. Thus,  $(\ell(\alpha),\ell(\beta))\in sim$.  
\qed
\end{proof}

\begin{definition}[representation mappings $\cl$ and $\lc$] \label{def:repmaps}
Let $\theta$ be a comtrace alphabet. 
\begin{enumerate}
 \item The mapping  $\cl:\quotient{\E_\theta^*}{\equiv_{\theta}} \rightarrow \LCT(\theta)$ is defined as
\[\cl(\mathbf{t}) \df \seq{\Sigma_{\mathbf{t}},\PO_{\mathbf{t}},\sq_{\mathbf{t}},\ell},\]
where the function $\ell:\Sigma_s\rightarrow E$ is defined in \secref{steps} and $S_{\mathbf{t}}=(\Sigma_{\mathbf{t}},\PO_{\mathbf{t}},\sq_{\mathbf{t}})$ is the so-structure defined by the comtrace $\mathbf{t}$ from \defref{s2sos}. 
 \item The mapping  $\lc:\LCT(\theta)\rightarrow \quotient{\E_\theta^*}{\equiv_{\theta}}$ is defined as
\[\lc\bigl((X,\PO,\sq,\lambda)\bigr) \df \Bset{map(\lambda,\Omega_\lhd)\mid \lhd \in ext\bigl((X,\PO,\sq)\bigr)}.\]
\EOD  
\end{enumerate}
\end{definition}

Intuitively, the mapping $\cl$ is used to convert a comtrace to  an lsos-comtrace while the mapping $\lc$ is used to transform an lsos-comtrace into a comtrace. The next theorem will show that $\cl$ and $\lc$ are valid representation mappings for $\quotient{\E_\theta^*}{\equiv_{\theta}}$ and  $\LCT(\theta)$.

\begin{theorem}[The 1$^{st}$ Representation Theorem]  \label{theo:rep}
Let $\theta$ be a comtrace alphabet.
\begin{enumerate}
\item For every $\mathbf{t}\in \quotient{\E_\theta^*}{\equiv_{\theta}}$, $\lc\circ \cl(\mathbf{t}) = \mathbf{t}$.
\item For every $T\in \LCT(\theta)$, $\cl\circ \lc(T) = T$.
\end{enumerate}
In other words, the following diagram commutes.
\begin{center}
      \begin{tikzpicture}[->,>=stealth',node distance=40mm]
        \tikzstyle{structure node}=
        [ minimum size=7mm,%
          circle,%
	text centered ]

        \node [structure node](A) {$\boldsymbol {\E^*/\equiv_{\theta}}$};
	\node [structure node] (B) [right of=A] {$\LCT(\theta)$};

        \path [thick] 
			 (A) edge[bend left] node[right,xshift= -6 mm,yshift=3mm]
                  {$\mathsf{ct2lct}$} (B) 
                  	 (A)	edge[loop left] node {$\mathsf{ id}_{\E^*/\equiv_{\theta}}$} (A)
                         (B) edge[bend left] node[right,xshift= -6 mm,yshift=-3mm]
                  {$\mathsf{lct2ct}$} (A)
                  	(B)	edge[loop right] node {$\mathsf{ id}_{\LCT(\theta)}$} (B); 

      \end{tikzpicture}
\end{center}
\end{theorem}

\begin{proof}\textbf{1. }The fact that ${\rm ran}(\cl)\subseteq \LCT(\theta)$ follows from \lemref{l3}. Now for a given $\mathbf{t}\in \quotient{\E_\theta^*}{\equiv_{\theta}}$,  we have $\cl(\mathbf{t})=(\Sigma_{\mathbf{t}},\PO_{\mathbf{t}},\sq_{\mathbf{t}},\ell)$. Thus, it follows that
\begin{align*}
\lc(\cl(\mathbf{t}))	& = \bset{\map(\ell,\Omega_\lhd)\mid \lhd \in ext(S_\mathbf{t})} 	& \\
		& = \bset{\map(\ell,\Omega_\lhd)\mid \lhd \in \set{\lhd_s\mid s\in \mathbf{t}}}	&\ttcomment{\text{by \theoref{com2sos}}}\\
		& = \bset{\map(\ell,\Omega_{\lhd_{s}})\mid s\in \mathbf{t}} &\\
	  	&= \mathbf{t}	&
\end{align*} 

\textbf{2. }Let $T_{0}=(\Sigma,\PO_{0},\sq_{0},\ell)$ be the canonical representation of $T$. Note that since $T_0\cong T$, we have 
\[\bset{\map(\ell,\Omega_\lhd)\mid \lhd \in ext(T_0)} = \bset{\map(\lambda,\Omega_\lhd)\mid \lhd \in ext(T)}.\]

Let $\Delta = \bset{\map(\ell,\Omega_\lhd)\mid \lhd \in ext(T_0)}$. We will show that $\Delta\in \quotient{\E_\theta^*}{\equiv_{\theta}}$ and  $\cl\bigl(\Delta\bigr) = \seq{T_0}$. Fix an arbitrary $u\in\Delta$, then by Lemmas \ref{lem:l1} and \ref{lem:l2}, it follows that $S_{[u]} = (\Sigma,\PO_{0},\sq_{0})$. Thus, from \theoref{com2sos}, we have 
\[\Delta =  \bset{\map(\ell,\Omega_\lhd)\mid \lhd \in ext(S_{[u]})} = [u].\] 
And the rest follows.  
\qed
\end{proof}

The theorem says that the mappings $\cl$ and $\lc$ are inverses of each other and hence are both \emph{bijective}.

\subsection{Representation theorem for lsos-comtraces and combined dependency graphs}

Recently, inspired by the dependency graph notion for Mazurkiewicz traces (cf. \cite[Chapter 2]{DR}), Kleijn and Koutny claimed without  a  proof that their \emph{combined dependency graph} notion is another alternative way to define comtraces \cite{KK08}. In this section, we will give a detailed proof of their claim. We will now recall the combined dependency graph definition.

\begin{definition}[combined dependency graph \cite{KK08}] \label{def:comdag}
A \emph{combined dependency graph} (\emph{cd-graph}) over a comtrace alphabet $\theta=(E,ser,sim)$ is (an lp-isomorphic class of) a finite labeled relational structure $D=\seq{X,\cau,\wcau,\lambda}$ such that $\lambda:X\rightarrow E$, the relations $\cau,\wcau$ are irreflexive, $D^{\lozenge}$ is an so-structure, and for all elements $\alpha\not=\beta$ of $X$,
\begin{center}
\tikzstyle{mybox} = [draw=black, very thick,
    rectangle, rounded corners, inner sep=5pt, inner ysep=10pt]

\begin{tikzpicture}
\node [mybox] (box){%
    \begin{minipage}{.5\textwidth}
\begin{itemize}
\item[] \textsf{CD1:\mbox{\hspace{5mm}}} $(\lambda(\alpha),\lambda(\beta))\not\in sim \implies \alpha \cau \beta \vee \beta \cau \alpha$
\item[] \textsf{CD2:\mbox{\hspace{5mm}}} $(\lambda(\alpha),\lambda(\beta))\not\in ser \implies \alpha \cau \beta \vee \beta \wcau \alpha$
\item[] \textsf{CD3:\mbox{\hspace{5mm}}} $\alpha \cau \beta \implies (\lambda(\alpha),\lambda(\beta))\not\in ser$
\item[] \textsf{CD4:\mbox{\hspace{5mm}}} $\alpha \wcau \beta \implies (\lambda(\beta),\lambda(\alpha))\not\in ser$
\end{itemize}
    \end{minipage}
};
\end{tikzpicture}%
\end{center}
We will write $\DCD(\theta)$ to denote the class of all cd-graphs over $\theta$.  \EOD
\end{definition}

Cd-graphs can be seen as a reduced graph-theoretic representation for lsos-comtraces, where some arcs that can be recovered using $\lozenge$-closure are omitted. It is interesting to observe that the non-serializable sets of a cd-graph are exactly the \emph{strongly connected components} of the directed graph $(X,\wcau)$ and can easily be found in time $O(|X|+|\wcau|)$ using any standard algorithm (cf. \cite[Section 22.5]{CLR}). \bigskip\\
\begin{minipage}{0.60\textwidth}\begin{remark} Cd-graphs were called \emph{dependence comdags} in \cite{KK08}. But this name could be misleading since the directed graph $(X,\wcau)$ is not necessarily acyclic. For example, the graph on the right is the cd-graph that corresponds to  the lsos-comtrace from \figref{f1}, but it is not acyclic  (here, we use the dashed edges to denote the relation $\wcau$ and the solid edges to denote \emph{only} the relation  $\cau$). Thus, we use the name ``combined dependency graph'' instead. \EOD\medskip\\ \end{remark}
\end{minipage}
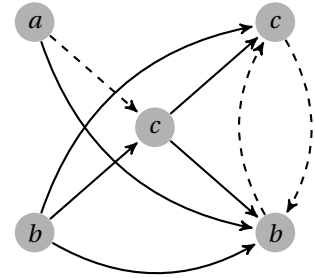
\begin{minipage}{0.4\textwidth}
\centering
\vspace{-3mm}
\begin{tikzpicture}[node distance=14mm,->,>=stealth']
        \tikzstyle{every node}=[circle,fill=black!30,minimum size=15pt,inner sep=0pt]

        \node (A) {$a$};
        \node (C) [below of=A,xshift=1.6cm] {$c$};
	\node (B) [above of=C,xshift=1.6cm] {$c$};
        \node (D) [below of=C,xshift=-1.6cm] {$b$};
        \node (E) [below of=C,xshift=1.6cm] {$b$};

        \path [thick,shorten >=1pt] 
        			(A)  edge[bend right] (E)
                         (C) edge (B) edge (E)
                         (D) edge (C) edge[bend right] (E) edge[bend left] (B);

        \path [dashed,thick,shorten >=1pt]
        			(A) edge (C)
			(B) edge[bend left] (E)
			(E) edge[bend left] (B);
\end{tikzpicture}			
\end{minipage}

We are going to show that the combined dependency graph notion is another correct alternative definition for comtraces. We will define several  representation mappings  that are needed for our proofs.

\begin{definition}[representation mappings $\cd$, $\dl$ and $\ld$] \label{def:drepmaps}
Let $\theta$ be a comtrace alphabet. 
\begin{enumerate}
 \item The mapping  $\cd:\quotient{\E_\theta^*}{\equiv_{\theta}} \rightarrow \DCD(\theta)$ is defined as
\[\cd(\mathbf{t}) \df (\Sigma_{\mathbf{t}},\PO_{u},\sq_{u},\ell),\]
where $u$ is any step sequence in $\mathbf{t}$ and $\PO_{u}$ and $\sq_{u}$ are defined as in \defref{s2inv}. 
 \item The mapping  $\dl:\DCD(\theta)\rightarrow \LCT(\theta)$ is defined as  $\dl(D) \df D^{\lozenge}$.
 \item The mapping  $\ld:\LCT(\theta)\rightarrow \DCD(\theta)$ is defined as
\[\ld(T)\df \cd \circ \lc (T) .\]  
\EOD
\end{enumerate}
\end{definition}

Before proceeding futher, we want to make sure that:
\begin{lemma}
\begin{enumerate}
\item The function $\dl:\DCD(\theta)\rightarrow \LCT(\theta)$ is well-defined.
\item The function $\cd:\quotient{\E_\theta^*}{\equiv_{\theta}} \rightarrow \DCD(\theta)$ is well-defined.
\end{enumerate}
\end{lemma}

\begin{proof} \textbf{1. }Given a cd-graph $D= \seq{X,\cau,\wcau,\lambda}\in \DCD(\theta)$, let $T= \seq{X,\PO,\sq,\lambda} = D^{\lozenge}$. We know from \defref{comdag} that  $(X,\PO,\sq)$ is an so-structure. It remains to show that $T$ is indeed an lsos-comtrace by verifying that $T$ satisfies the conditions \textsf{LC1}-\textsf{LC5} of \defref{lcomtrace}. 

We first verify \textsf{LC1}. Suppose for contradiction that there two distinct $\sq$-cycle equivalence classes $[\alpha],[\beta]\subset X$ satisfying $[\alpha](\imm{\hat{\sq}}\cap \hat{\PO})[\beta]$ but $\lambda\bigl[[\alpha]\bigr]\times \lambda\bigl[[\beta]\bigr] \subseteq ser$. Clearly, this implies that $\alpha \PO \beta$, and thus by the $\lozenge$-closure definition, $\beta$ is reachable from $\alpha$ on the directed graph $G=(X,\bcau)$, where $\bcau = \cau\cup \wcau$. Now we consider a shortest path $P$ 
\[\alpha = \delta_1 \bcau\delta_2 \bcau\ldots\bcau \delta_{k-1}\bcau \delta_k = \beta\]
on $G$ that connects $\alpha$ to $\beta$. Our strategy is to show that there exist two consecutive elements $\delta_i$ and $\delta_{i+1}$ on $P$ such that $\delta_i\in [\alpha]$ and $\delta_{i+1}\in [\beta]$ and $(\lambda(\delta_{i}),\lambda(\delta_{i+1}))\not \in ser$, which contradicts with  $\lambda([\alpha])\times \lambda([\beta]) \subseteq ser$. By \textsf{CD3} it suffices to show the following claim.
\begin{center}
\tikzstyle{mybox} = [draw=black, very thick,
    rectangle, rounded corners, inner sep=5pt, inner ysep=5pt]
\begin{tikzpicture}
\node [mybox] (box){%
    \begin{minipage}{0.67\textwidth}\centering
        \textbf{Claim:} There are two consecutive elements $\delta_i$ and $\delta_{i+1}$ on $P$ such 	that \\
	$\delta_i\in [\alpha]$ and $\delta_{i+1}\in [\beta]$ and  $\delta_i \longrightarrow \delta_{i+1}$.
    \end{minipage}
};
\end{tikzpicture}%
\end{center}
We will prove this claim by induction on the number of elements on the path $P$, where $k\ge 2$. 
\begin{itemize}
\item[]\textbf{Base case:} when $k=2$, then $\alpha \bcau \beta$. Since $[\alpha](\imm{\hat{\sq}}\cap \hat{\PO})[\beta]$, we have $\alpha \cau \beta$.
\item[]\textbf{Inductive case:} when $k>2$, we consider  first two elements on the path $\delta_1$ and  $\delta_2$. If $\delta_1\in [\alpha]$ and $\delta_2\in [\beta]$, then by the assumption $[\alpha](\imm{\hat{\sq}}\cap \hat{\PO})[\beta]$,  it must be the case that  $\delta_1 \cau \delta_2$. Otherwise, we have $\delta_2\not\in [\alpha]\cup [\beta]$ or $\bset{\delta_1,\delta_2}\subseteq [\alpha]$. For the first case, we get $[\alpha]\hat{\sq}[\delta_2]\hat{\sq}[\beta]$, which contradicts that $[\alpha]\imm{\hat{\sq}}[\beta]$. For the latter case, we can apply the induction hypothesis on the path $\delta_2 \bcau\ldots\bcau \delta_{k-1}\bcau \delta_k$.
\end{itemize}
\textsf{LC2} and \textsf{LC3} can also be shown similarly using a ``shortest path'' argument as above. These proofs are easier since we only need to consider paths with edges in $\wcau$.  \textsf{LC4} and \textsf{LC5} easily follow from the fact that the cd-graph $D$ satisfies  \textsf{CD1} and \textsf{CD2}.\\
 
\textbf{2. }By the proof of \cite[Lemma 4.7]{JK95}, for any two step sequences $t$ and $u$ in $\E_\theta^*$, we have $u \eqb t$ if and only if $\cd([u]) = \cd([t])$. This ensures that  $\cd(\mathbf{t})$ gives us the same cd-graph no matter how we choose the step sequence $u \in \mathbf{t}$.  It is also not hard to check that the range of $\cd$ consists only of cd-graphs over $\theta$.  Thus the mapping $\cd$ is well-defined. 
\qed
\end{proof}

\begin{lemma} \label{lem:dlinj}
The mapping  $\dl:\DCD(\theta)\rightarrow \LCT(\theta)$ is injective.
\end{lemma}
\begin{proof} Assume that $D_1,D_2 \in \DCD(\theta)$, such that 
\[\dl(D_1)=\dl(D_2) = T= [X,\PO,\sq,\lambda].\] 
Since $\lozenge$-closure does not change the labeling function, we can assume  $D_i= [X,\cau_i,\wcau_i,\lambda]$ and $(X,\cau_i,\wcau_i)^{\lozenge} = (X,\PO,\sq)$. We want to  show that $(X,\cau_1,\wcau_1)\subseteq(X,\cau_2,\wcau_2)$.

($\cau_1\;\subseteq\;\cau_2$): Let $\alpha,\beta\in X$ such that $\alpha \cau_1 \beta$. Suppose for  a contradiction that $\neg (\alpha \cau_2 \beta)$. Since $\alpha \cau_1 \beta$, by \textsf{CD3},  $(\lambda(\alpha),\lambda(\beta))\not\in ser$. Thus, by \textsf{CD2}, we must have $\beta \wcau_2 \alpha$. But since $(X,\cau_i,\wcau_i)^{\lozenge} = (X,\PO,\sq)$, it follows that $(X,\cau_i,\wcau_i) \subseteq (X,\PO,\sq)$ (by \propref{so-cl}). Thus, $\alpha\PO\beta$ and $\beta \sq \alpha$, a contradiction.

($\wcau_1\;\subseteq\;\wcau_2$): Can be proved similarly.

By reversing the role of $D_1$ and $D_2$, we have $(X,\cau_1,\wcau_1)\supseteq(X,\cau_2,\wcau_2)$. Thus, $D_1=D_2$.
\qed
\end{proof}

We are now ready to show the following representation theorem which ensures that $\ld$ and $\dl$ are valid representation mappings for $\LCT(\theta)$ and $\DCD(\theta)$.

\begin{theorem}[The 2$^{nd}$ Representation Theorem] \label{theo:deprep}
Let $\theta$ be a comtrace alphabet.
\begin{enumerate}
\item For every $T\in \LCT(\theta)$, $\dl \circ \ld(T) = T$. 
\item For every $D\in \DCD(\theta)$, $\ld \circ \dl(D) = D$. 
\end{enumerate}
In other words, the following diagram commutes.
\begin{center}
      \begin{tikzpicture}[->,>=stealth',node distance=40mm]
        \tikzstyle{struct node}=
        [%
          minimum size=7mm,%
          circle,%
          thick,%
	text centered
        ]

        \node [struct node](A) {$\LCT(\theta)$};
	\node [struct node] (B) [right of=A] {$\DCD(\theta)$};

        \path [thick] 
			 (A) edge[bend left] node[right,xshift= -6 mm,yshift=3mm]
                  {$\mathsf{lct2dep}$} (B) 
                  	 (A)	edge[loop left] node {$\mathsf{ id}_{\LCT(\theta)}$} (A)
                         (B) edge[bend left] node[right,xshift= -6 mm,yshift=-3mm]
                  {$\mathsf{dep2lct}$} (A)
                  	(B)	edge[loop right] node {$\mathsf{ id}_{\DCD(\theta)}$} (B); 

      \end{tikzpicture}
\end{center}
\end{theorem}
\begin{proof}\textbf{1. } Let $T\in \LCT(\theta)$ and let $D = \ld(T)$. Suppose for a contradiction that $Q = \dl \circ \ld(T)$ and $Q\not= T$. Since $\ld = \cd\circ \lc$, if we let $\mathbf{t}= \lc(T)$, then  $Q = \dl\circ \cd (\mathbf{t})\not= T$. Thus, we have shown that $\mathbf{t}=\lc(T)$ and by the way we construct so-structure from comtraces, we also have $\cl(\mathbf{t})=\dl\circ \cd (\mathbf{t})=Q$. Since $Q\not= T$, we have $\cl \circ \lc (T) \not= T$, which  contradicts \theoref{rep} (2).

\textbf{2. }Let $D \in \DCD(\theta)$ and $T = \dl(D)$. Suppose for a contradiction that $E =\ld \circ \dl(D)$ and $E\not = D$. From \textbf{1.}, we know that  $\dl \circ \ld (T) = T$, and thus it must be the case that $\dl(E)= T$. Hence, we have $\dl(E)= T = \dl(D)$, but $E\not=D$, which contradicts the injectivity of $\dl$ from \lemref{dlinj}.
\qed
\end{proof}

This theorem shows that both $\ld$ and $\dl$ are bijective. Note that we do not need to prove another representation theorem for cd-graphs and comtraces since their representation mappings are simply the composition of the representation mappings from Theorems \ref{theo:rep} and \ref{theo:deprep}. In other words, we have shown that the following diagram commutes.
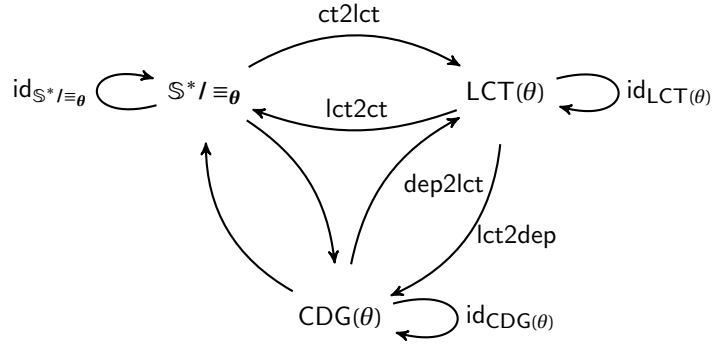
\begin{center}
      \begin{tikzpicture}[->,>=stealth',node distance=40mm]
        \tikzstyle{struct node}=
        [%
          	minimum size=5mm,%
          	circle,%
        	thick,%
		text centered
        ]

        \node [struct node](A) {$\boldsymbol{\E^*/\equiv_{\theta}}$};
	\node [struct node] (B) [right of=A] {$\LCT(\theta)$};
	\node [struct node] (C) [below of=B,xshift=-22mm,yshift=10mm] {$\DCD(\theta)$};
	
        \path [thick] 
			 (A) edge[bend left] node[right,xshift= -6 mm,yshift=3mm]
                  {$\mathsf{ct2lct}$} (B) 
                         (B) edge[bend left=20pt] node[right,xshift= -5 mm,yshift=3mm]
                  {$\mathsf{lct2ct}$} (A)
                  	 (A)	edge[loop left] node {$\mathsf{ id}_{\boldsymbol {\E^*/\equiv_{\theta}}}$} (A)
                  
			(B) edge[bend left=30pt] node[right,xshift= 0 mm,yshift=0mm]
                  {$\mathsf{lct2dep}$} (C) 
			(C) edge[bend left=25pt] node[left,xshift= 14 mm,yshift=-1mm]
                  {$\mathsf{dep2lct}$} (B) 

                  	(B)	edge[loop right] node {$\mathsf{ id}_{\LCT(\theta)}$} (B) 

			(C) edge[bend left] (A) 
			(A) edge[bend left=25pt] (C) 

                  	(C)	edge[loop right] node {$\mathsf{ id}_{\DCD(\theta)}$} (C); 

      \end{tikzpicture}
\end{center}
In \secref{op}, after constructing suitable composition operators for lsos-comtraces and cd-graphs, we will show that the representation mappings in this diagram are indeed monoid isomorphisms. Thus, lsos-comtraces and cd-graphs are equivalent representations for comtraces. 

\section{Composition operators \label{sec:op}}
For a comtrace monoid $(\E^*/\!\!\equiv_{\theta},\circledast,[\epsilon])$, the comtrace operator $\_\circledast\_$ was defined as $[r]\circledast[t] = [r\ast t]$. We will construct analogous composition operators for lsos-comtraces and cd-graphs. We will then show that lsos-comtraces (cd-graphs) over a comtrace alphabet $\theta$ together with its composition operator form a monoid isomorphic to the comtrace monoid $(\E^*/\!\!\equiv_{\theta},\circledast,[\epsilon])$. In other words, we need to show that the mappings from Theorems \ref{theo:rep} and \ref{theo:deprep} are compatible with the corresponding monoid operators. 

\subsection{Monoid of lsos-comtraces}
Given two sets $X_{1}$ and $X_{2}$, we  write $X_{1}\uplus X_{2}$ to denote the \emph{disjoint union} of $X_{1}$ and $X_{2}$. Such disjoint union can be easily obtained by renaming the elements in $X_{1}$ and $X_{2}$ so that $X_{1}\cap X_{2} =\emptyset$. We define the lsos-comtrace composition operator as follows.

\begin{definition}[composition of lsos-comtraces] \label{def:op1}
Let $T_{1}$ and $T_{2}$ be lsos-comtraces over an alphabet $\theta=(E,sim,ser)$, where $T_{i}=\seq{X_{i},\prec_{i},\sqsubset_{i},\lambda_{i}}$. The \emph{composition} $T_{1}\odot T_{2}$ is defined as (an lp-isomorphic class of) a labeled so-structure $\seq{X,\prec,\sqsubset,\lambda}$ such that $X=X_{1}\uplus X_{2}$, $\lambda = \lambda_{1}\uplus \lambda_{2}$, and $(X,\prec,\sqsubset) = \left(X,\PO_{\langle1,2\rangle},\sq_{\langle1,2\rangle}\right)^{\lozenge}$, where
\begin{align*}
\PO_{\langle1,2\rangle} &\;=\; \PO_{1}\cup \PO_{2}\cup\; \bset{(\alpha,\beta)\in X_{1}\times X_{2}\mid (\lambda(\alpha),\lambda(\beta))\not\in ser}\\
\sq_{\langle1,2\rangle} &\; =\; \sq_{1}\cup \sq_{2}\cup\; \bset{(\alpha,\beta)\in X_{1}\times X_{2}\mid (\lambda(\beta),\lambda(\alpha))\not\in ser} 
\end{align*}
\EOD 
\end{definition}

The operator $\_\odot\_$ is well-defined since  we can easily check that:
\begin{proposition}
For every $T_1,T_2\in \LCT(\theta)$, $T_1\odot T_2 \in \LCT(\theta)$.   
\end{proposition}

We will next show that this composition operator $\_\odot\_$ properly corresponds to the operator $\_\circledast\_$  of the comtrace monoid over $\theta$.

\begin{lemma}  \label{lem:hom1}
Let $\theta$ be a comtrace alphabet. Then  
\begin{enumerate}
\item For every $R,T \in \LCT(\theta)$, $\lc(R \odot T) = \lc(R) \circledast \lc(T).$

\item For every $\mathbf{r},\mathbf{t} \in \quotient{\E_\theta^*}{\equiv_{\theta}}$,
$\cl(\mathbf{r}\circledast \mathbf{t})= \cl(\mathbf{r})\odot\cl(\mathbf{t}).$
\end{enumerate}
\end{lemma}
\begin{proof}\textbf{1. }Without loss of generality, we can assume that $R=\seq{X_1,\PO_1,\sq_1,\lambda_1}$, $T=\seq{X_2,\PO_2,\sq_2,\lambda_2}$ and $Q=R \odot T=\seq{X_1\uplus X_1,\PO,\sq,\lambda}$, where $\lambda=\lambda_1\uplus\lambda_2$. We can pick any $\lhd_1\in ext(R)$ and $\lhd_2\in ext(T)$. Then observe that the stratified order $\lhd$ satisfying $\Omega_{\lhd} = \Omega_{\lhd_1}\ast \Omega_{\lhd_2}$ is also a stratified extension of $Q$. Thus, by \theoref{rep}, we get
\[\lc(R) \circledast \lc(T) = [\map(\lambda_1,\lhd_1)]\circledast [\map(\lambda_2,\lhd_2)] =  [\map(\lambda,\lhd)] = \lc(Q)\] as desired.

\textbf{2. } Without loss of generality, we assume that $\mathbf{r}= [r]$, $\mathbf{t}=[t]$ and $\mathbf{q}=[q] = \mathbf{r}\circledast \mathbf{t}$, where $q = r \ast t$. By reindexing $\Sigma_{\mathbf{t}}$ appropriately, we can also assume that $\Sigma_{\mathbf{q}} = \Sigma_{\mathbf{r}}\uplus \Sigma_{\mathbf{t}}$. Under these assumptions, let 
\begin{align*}
T_{1} &= \seq{\Sigma_{\mathbf{r}},\PO_{\mathbf{r}},\sq_{\mathbf{r}},\ell_{1}}=\cl(\mathbf{r}),\\ 
T_{2} &= \seq{\Sigma_{\mathbf{t}},\PO_{\mathbf{t}},\sq_{\mathbf{t}},\ell_{2}}=\cl(\mathbf{t}),\\
T &= \seq{\Sigma_{\mathbf{q}},\PO_{\mathbf{q}},\sq_{\mathbf{q}},\ell}=\cl(\mathbf{q}),
\end{align*}
where $\ell=\ell_{1}\uplus \ell_{2}$ is simply the standard labeling functions. It will now suffice to show that $T_{1}\odot T_{2} = T$.

($\subseteq$): Let $T_{1}\odot T_{2} = (\Sigma_{\mathbf{r}}\uplus \Sigma_{\mathbf{t}},\PO_{\langle \mathbf{r},\mathbf{t}\rangle},\sq_{\langle \mathbf{r},\mathbf{t}\rangle},l)^{\lozenge}$. By Definitions \ref{def:s2inv} and \ref{def:s2sos}, we have
\begin{align*}
 \PO_{\langle \mathbf{r},\mathbf{t}\rangle}&\;=\; \PO_{\mathbf{r}}\cup \PO_{\mathbf{t}}\cup\; \bset{(\alpha,\beta)\in \Sigma_{\mathbf{r}}\times \Sigma_{\mathbf{t}}\mid (\lambda(\alpha),\lambda(\beta))\not\in ser} \subseteq\, \PO_{\mathbf{q}}\\
\sq_{\langle \mathbf{r},\mathbf{t}\rangle}&\;=\; \sq_{\mathbf{r}}\cup \sq_{\mathbf{t}}\cup\; \bset{(\alpha,\beta)\in \Sigma_{\mathbf{r}}\times \Sigma_{\mathbf{t}}\mid (\lambda(\beta),\lambda(\alpha))\not\in ser}\subseteq\, \sq_{\mathbf{q}}
\end{align*}
Thus, from \propref{so-cl} (5),  we have $(\Sigma_{\mathbf{r}}\uplus \Sigma_{\mathbf{t}},\PO_{\langle \mathbf{r},\mathbf{t}\rangle},\sq_{\langle \mathbf{r},\mathbf{t}\rangle},\ell)^{\lozenge}\subseteq  (\Sigma_{\mathbf{q}},\PO_{\mathbf{q}},\sq_{\mathbf{q}},\ell)$.

($\supseteq$): By Definitions \ref{def:s2inv} and \ref{def:s2sos}, we have $\PO_{q}\subseteq\, \PO_{\langle \mathbf{r},\mathbf{t}\rangle}$ and $\sq_{q}\subseteq\, \sq_{\langle \mathbf{r},\mathbf{t}\rangle}$. 
From \defref{op1}, we already know that  $ \sq_{\langle \mathbf{r},\mathbf{t}\rangle}$ is irreflexive since $\sq_{\bf{t}}$ and $\sq_{\bf{r}}$ are irreflexive. Thus, by \propref{so-cl} (1), 
\[(\Sigma_{\mathbf{r}}\uplus \Sigma_{\mathbf{t}},\PO_{\langle \mathbf{r},\mathbf{t}\rangle},\sq_{\langle \mathbf{r},\mathbf{t}\rangle})\subseteq (\Sigma_{\mathbf{r}}\uplus \Sigma_{\mathbf{t}},\PO_{\langle \mathbf{r},\mathbf{t}\rangle},\sq_{\langle \mathbf{r},\mathbf{t}\rangle})^{\lozenge}\]
Hence, we have $(\Sigma_{\mathbf{q}},\PO_{q},\sq_{q},\ell)\subseteq (\Sigma_{\mathbf{r}}\uplus \Sigma_{\mathbf{t}},\PO_{\langle \mathbf{r},\mathbf{t}\rangle},\sq_{\langle \mathbf{r},\mathbf{t}\rangle})^{\lozenge}$. Thus, from \propref{so-cl} (5), 
\[T=(\Sigma_{\mathbf{q}},\PO_{\mathbf{q}},\sq_{\mathbf{q}},\ell) = (\Sigma_{\mathbf{q}},\PO_{q},\sq_{q},\ell)^{\lozenge}\subseteq (\Sigma_{\mathbf{r}}\uplus \Sigma_{\mathbf{t}},\PO_{\langle \mathbf{r},\mathbf{t}\rangle},\sq_{\langle \mathbf{r},\mathbf{t}\rangle},\ell)^{\lozenge} = T_{1}\odot T_{2}.\]
Thus, we have shown that $T_{1}\odot T_{2} = T$. 
\qed
\end{proof}

Let $\mathbb{I}$ denote the lp-isomorphic class  $[\emptyset,\emptyset,\emptyset,\emptyset]$. Then observe that $\cl([\epsilon]) = \mathbb{I}$ and $\lc(\mathbb{I}) = [\epsilon]$. By \lemref{hom1} and \theoref{rep}, it follows that the structures $(\LCT(\theta),\odot, \mathbb{I})$ and $(\quotient{\E_\theta^*}{\equiv_{\theta}},\circledast,[\epsilon])$ are isomorphic under the isomorphisms $\cl:\quotient{\E_\theta^*}{\equiv_{\theta}} \rightarrow \LCT(\theta)$ and $\lc:\LCT(\theta)\rightarrow \quotient{\E_\theta^*}{\equiv_{\theta}}$. Thus, the triple $(\LCT(\theta),\odot , \mathbb{I})$ is also a monoid. We can summarize these facts in the following theorem:
\begin{theorem} \label{theo:iso1}
The mappings $\cl$ and $\lc$ are monoid isomorphisms between $(\quotient{\E_\theta^*}{\equiv_{\theta}},\circledast,[\epsilon])$ and  $(\LCT(\theta),\odot, \mathbb{I})$.  
\end{theorem}

\subsection{Monoid of cd-graphs}
Similarly to the previous section, for a given a comtrace alphabet, we can also define a composition operator for cd-graphs.
\begin{definition}[composition of cd-graphs] Let $D_{1}$ and $D_{2}$ be cd-graphs over an alphabet $\theta=(E,sim,ser)$, where $D_{i}=\seq{X_{i},\cau_{i},\wcau_{i},\lambda_{i}}$. The \emph{composition} $D_{1}\circledcirc D_{2}$  is defined as (an lp-isomorphic class of) a labeled so-structure $\seq{X,\cau,\wcau,\lambda}$ such that $X=X_{1}\uplus X_{2}$, $\lambda = \lambda_{1}\uplus \lambda_{2}$, and 
\begin{align*}
\cau\;&=\; \cau_{1}\cup \cau_{2}\cup\; \set{(\alpha,\beta)\in X_{1}\times X_{2}\mid (\lambda(\alpha),\lambda(\beta))\not\in ser}\\ 
\wcau\;&=\; \wcau_{1}\cup \wcau_{2}\cup\; \set{(\alpha,\beta)\in X_{1}\times X_{2}\mid (\lambda(\beta),\lambda(\alpha))\not\in ser}
\end{align*}
\EOD
\end{definition}

The operator $\_\circledcirc\_$ is well-defined since  we can easily check that:

\begin{proposition}
For every $D_1,D_2\in \DCD(\theta)$, $D_1\circledcirc D_2 \in \DCD(\theta)$.   
\end{proposition}

Using techniques similar to the proofs of \lemref{hom1} and \theoref{deprep}, it is not hard to show the following lemma.

\begin{lemma}\label{lem:hom2}
Let $\theta$ be a comtrace alphabet. Then  
\begin{enumerate}
\item For every $R,T \in \LCT(\theta)$, $\ld(R \odot T) = \ld(R) \circledcirc \ld(T).$
\item For every $D,E \in \DCD(\theta)$, $\dl(D \circledcirc E) = \dl(D) \odot \dl(E).$  
\end{enumerate}
\end{lemma}

Putting the preceding lemma  and \theoref{deprep} together, we conclude:
\begin{theorem} \label{theo:iso2}
The mappings $\ld$ and $\dl$ are monoid isomorphisms between $(\LCT(\theta),\odot, \mathbb{I})$ and $(\DCD(\theta),\circledcirc, \mathbb{I})$.  
\end{theorem}

By composing the isomorphisms from Theorems \ref{theo:iso1} and \ref{theo:iso2}, we have:

\begin{corollary}
The monoids $(\quotient{\E_\theta^*}{\equiv_{\theta}},\circledast,[\epsilon])$  and $(\DCD(\theta),\circledcirc, \mathbb{I})$ are isomorphic.  
\end{corollary}

\section{Conclusion \label{sec:conc}}

The simple yet useful construction we used extensively in this paper is to build a quotient so-structure modulo the $\sq$-cycle equivalence relation. Intuitively, each $\sq$-cycle equivalence class consists of  events that must be executed simultaneously with one another and hence can be seen as a single ``composite event''.  The resulting quotient so-structure  is technically  easier to handle since both relations of the quotient so-structure are acyclic. From this construction, we were able to give a labeled so-structure definition of comtraces analogous to the labeled poset definition of traces. 

We have also formally shown that the quotient monoid of comtraces, the monoid of lsos-comtraces and the monoid of cd-graphs over the same comtrace alphabet are  isomorphic by constructing monoid isomorphisms between them. These three models are  formal linguistic,  order-theoretic, and  graph-theoretic respectively, which allows us to apply a variety of tools and techniques. We believe the ability  to conceptualize on three alternative representations is the main advantage of trace theory in general.

An immediate future task is to develop a framework similar to the one in this paper for \emph{generalized comtraces}, proposed and developed in \cite{Le08,JL11}. Generalized comtraces extend comtraces with the ability to model events that can be executed \emph{earlier than or later than but never simultaneously}.  We believe that the quotient so-structure technique developed in this paper can be used to simplify some proofs in \cite{JL11}.

The labeled so-structure definition of comtraces can easily be extended to define \emph{infinite comtraces} to model nonterminating concurrent processes, and thus it would be interesting to generalize the results in  \cite{Gas90,Die91} for comtraces. It is also promising to use lsos-comtraces and cd-graphs  to develop logics for comtraces similarly to what have been done for traces (see, e.g., \cite{TW02,Wal02,DG06a,DG06b,GK07,DHK07}). 

\section*{Acknowledgment}
I would like to thank the anonymous referees who pointed out several serious typos and mistakes in the earlier version of this paper.  Their comments also helped improving the presentation of this paper substantially. I am grateful to the Mathematics Institute of Warsaw University and the Theoretical Computer Science Group of Jagiellonian University for their generous supports during my visits in 2009. It was during these visits that the ideas from this paper emerge.  I am indebted to Steve Cook, Gabriel Juh\'as, Ryszard Janicki, Jetty Kleijn, Maciej Koutny and  Yuli Ye  for their suggestions and encouragement. This research is financially  supported by the Ontario Graduate Scholarship and the Natural Sciences and Engineering Research Council of Canada.

\bibliographystyle{plain}
\bibliography{comtraces}
\end{document}